\documentclass[reqno]{amsart}
\usepackage{graphicx}

\newtheorem{lemma}{Lemma}
\newtheorem{theorem}{Theorem}

\newtheorem{remark}{Remark}
\newtheorem{definition}{Definition}

\newcommand{\bee}{\begin{eqnarray}}
\newcommand{\eee}{\end{eqnarray}}
\newcommand{\be}{\begin{eqnarray*}}
\newcommand{\ee}{\end{eqnarray*}}
\newcommand{\R}{{\mathbb R}}
\newcommand{\N}{{\mathbb N}}

\newcommand{\C}{{\mathbb C}}

\newcommand{\W}{{\mathcal W}}
\newcommand{\D}{{\mathcal D}}
\newcommand{\h}{h}

\begin{document}
 
 \title [Quantum forced oscillator via Wigner transform]{Quantum forced oscillator via Wigner transform}
 
 \author {
 Andrea Sacchetti
 }

\address {
Department of Physics, Informatics and Mathematics, University of Modena and Reggio Emilia, Modena, Italy.
}

\email {andrea.sacchetti@unimore.it}

\date {\today}

\thanks {This work is partially supported by GNFM-INdAM and by the UniMoRe-FIM project ``Modelli e metodi della Fisica Matematica''.}

\begin {abstract} 
In this paper we review the basic results concerning the Wigner transform and then we completely solve the quantum forced harmonic/inverted oscillator in such a framework; 
eventually, the tunnel effect for the forced inverted oscillator is discussed.
\end{abstract}

\maketitle

\section {Introduction} \label {Sez1}

The formulation of quantum mechanics in the phase space has been a topic of great interest since Wigner's groundbreaking paper \cite {W}. \ In such a treatment the explanation 
of the quantum-classical correspondence seems to be very spontaneous and many attempts have been done with the aim to shed light on quantum mechanics from a classical 
perspective \cite {Berry,Heller,MMF}. 

Let $\psi (x , t)$, $x \in \R^d$, be the solution to the time-dependent Schr\"odinger equation 
\bee
\left \{
\begin {array}{l}
i \hbar \frac {\partial \psi}{\partial t} = - \frac {\hbar^2}{2m} \Delta \psi + V \psi \\
\psi (x,0) =\psi_0 (x) 
\end {array}
\right. \, , \ \psi (\cdot , t ) \in L^2 (\R^d , dx)\, . \label {Eq1}
\eee
The Wigner (or Wigner-Slizard\footnote {In fact, the expression of such a transformation was introduced by Wigner; however, it was found by L. Slizard and E. Wigner for a different 
purpose some years before the publication of the paper \cite {W}.}) transform defined in the phase-space $(x,\xi ) \in \R^{2d}$ as follows 
\be
W^\hbar (x,\xi ,t) = \frac {1}{[2\pi]^d} \int_{\R^d} e^{-i \xi \cdot y } \psi \left ( x + \frac 12 \hbar y ,t \right ) \overline {\psi \left ( x - \frac 12 \hbar y ,t 
\right )} d y 
\ee
solves to the Wigner equation 
\bee
\frac {\partial W^\hbar}{\partial t} = - \frac {1}{m} \xi \cdot \nabla_x W^\hbar + \nabla_x V \cdot \nabla_\xi W^\hbar + {\mathcal R} (W^\hbar ) 
= \left \{ W^\hbar , h \right \} + {\mathcal R} (W^\hbar ) 
\label {Eq2}
\eee
where $h = \frac {1}{2m} \xi^2 + V$ is the classical Hamiltonian function, $\{ \cdot , \cdot \}$ denotes the Poisson's brackets and ${\mathcal R} (W^\hbar )$ is a 
remainder term depending on the unknown function $W^\hbar$ that formally goes to zero in the semiclassical limit $\hbar \to 0$. \ Thus, as Wigner himself pointed out, 
in the semiclassical limit the function $W^\hbar$ solves the classical Liouville equation $\frac {\partial W^\hbar}{\partial t} = \left \{ W^\hbar , h \right \}$.

Furthermore, remembering that in the standard notation of quantum mechanics the probability density in position space is given 
by $|\psi (x,t)|^2$ then another advantage of using the Wigner representation is related to the evaluation of expectation values for 
physical observables; in fact
\be
|\psi (x,t)|^2 = \int_{\R^d} W^\hbar (x,\xi ,t) d\xi 
\ee
and then the probability density can be obtained by means of the solution $W^\hbar$ to the Wigner equation (\ref {Eq2}). \ Similarly, the probability density in the momentum space follows, too.

Thus, at first glance, it seems that the program of studying the semiclassical limit of quantum mechanics by means of the classical Liouville equation is 
very promising; unfortunately, the residual term ${\mathcal R}$ usually contains derivatives of the unknown function $W^\hbar$ of higher than prime order, 
and so rigorous treatment of the complete equation (\ref {Eq2}) is a formidable task. \ Many attempts, both theoretical and numerical, have been done 
\cite {AE,CG,DM,FBFG,MMP,M,PRB,P,SS,VPSM} in order to find an approximate solution to (\ref {Eq2}) but some gaps in the theory are still open when the potential $V$ is 
not a quadratic function with respect to the spatial variable. \ Indeed, equation (\ref {Eq2}) is in general impossible to be analytically solved and only some 
approximated solutions $\tilde w$ may be obtained; in such a case a first problem occurs: {\it is $\tilde w$ the Wigner transform of a square integrable 
wave-function $\tilde \psi$?} and secondly: {\it if $\tilde w$ approximates, with respect to some norm, the exact solution $w$ to (\ref {Eq2}) corresponding to the 
wave-function $\psi $, then does $\tilde \psi$ approximate $\psi$ in the $L^2$ space?} \ As discussed in Remark \ref {Nota6}, as far as we known a fully satisfactory 
answer to these crucial questions is not yet given. 

However, when $V$ is a quadratic function with respect to the spatial variable then the remainder term $\mathcal R$ in (\ref {Eq2}) is exactly zero and thus the Wigner equation exactly reduces 
to the classical Liouville equation that admits an explicit solution. \ The case of quadratic potential 
basically reduces to the harmonic/inverted oscillator and in such a case equation (\ref {Eq1}) has been widely studied  obtaining the explicit expression of the 
propagator using, for instance, the Feynman path integrals \cite {FH,Hira,MY,Mo}, or by solving the classical Liouville equation \cite {BV,HSADV}, or by using the 
Heisenberg picture \cite {Barton}; see also \cite {Teta} for a review. \ A special attention has been paid in order to study the tunnel effect in the case of inverted oscillator \cite {BV,Barton,HSADV}.

In this paper we briefly review the Wigner transform and the dynamics of the wave-function in the Wigner representation. \ For sake of simplicity we restrict our 
attention to the one-dimensional case; the extension to higher dimension is quite simple (see, e.g., \cite {FM}). \ Then we focus our attention to the case of 
potentials of second degree with respect to the spatial variable of the form $V(x,t)= \gamma x^2 + Q(t) x$, where $\gamma $ is a constant factor (in principle, one 
could similarly treat the case of $\gamma (t)$ depending on $t$ but we don't dwell here on this problem). \ In such a case an explicit solution $W^\hbar (x,\xi ,t)$ to 
the classical Liouville equation is given. \ Then we apply our results to the particular case where the initial wave-function $\psi_0 (x)$ has a Gaussian shape; in 
such a case the tunnel effect for the inverted oscillator is studied in detail. \ We should remark that most of the results contained in this paper have been already 
obtained in previous works; our aim is to collect all these results in a self-consisting paper adding some new results concerning, for example, the continuity of the 
Wigner transform or the explicit solution in the Wigner representation for the forced harmonic/inverted oscillator.

The paper is organized as follows. In Section \ref {Sez2} we recall the definition of Wigner transform and we collect some basic properties in Subsection \ref {Sez2_1}; these 
results are not new and they have been already proved in the papers cited there; however, for the sake of self-consistness, a brief proof is also given. \ In 
Subsection \ref {Sez2_2} we discuss the problem of the invertibility of the Wigner transform; in fact, an important and, under some aspects, open problem consists 
in checking if a real-valued function defined in the phase-space is the Wigner transform of a pure/mixed state or not. \ In Subsection \ref {Sez2_3} we 
discuss a problem that usually is not treated: the continuity of the Wigner transform. \ In Subsection \ref {Sez2_4} we collect some important examples of 
Wigner transform of wave-functions $\psi (x,t)$.

In Section \ref {Sez3} we study the Schr\"odinger equation in the Wigner representation. \ As we have already discussed above the resulting equation is, in general, 
quite hard to be treated because the remainder term ${\mathcal R}(W^\hbar )$ consists of a formal sum of derivatives of the unknown function $W^\hbar$ of any 
order. \ In Subsection \ref {Sez3_2} we consider a couple of special Schr\"odinger equations: the case where the potential is given by a Dirac's delta and the case of 
the nonlinear Schr\"odinger equation; in both cases the resulting equation in the phase-space is an integro-differential equation.

In Section \ref {Sez4} we consider the case of the forced harmonic/inverted oscillator. \ Such a problem may be treated by solving the Schr\"odinger equation, as 
done in standard textbooks (see. e.g., \cite {Teta}). \ In fact, in the case of forced oscillator one could reduce it to the unforced one by means of a suitable 
change of variable as done by Husimi \cite {Husimi}, and then treats the latter model by means of the same arguments discussed in Remark \ref {Nota4} because in 
the case of harmonic potential equation (\ref {Eq21}) simply reduces to (\ref {Eq22}). \ However, as we can see in Subsection  \ref {Sez4_2}, the treatment in 
the phase-space is much more simply because the classical Liouville equation is reduced to a simple system of two ODEs. 

Finally, in Section \ref {Sez5} we apply the results obtained in the previous Section to the study of the dynamics of a Gaussian wave-function. \ The detailed study of the tunnel effect is given in Subsection \ref {Sez5_4}.

A couple of short appendices with technical calculations complete the paper.

\section {The Wigner transform} \label {Sez2}

\subsection {Definition and main properties} \label {Sez2_1}

Let $\varphi , \, \phi \in L^2 (\R ,dx)$, then we define a new function in phase space $(x,\xi )\in \R^2$ as follows
\be
\widetilde W^\hbar (x,\xi ):= \left [ \widetilde {\W^\hbar} (\varphi , \phi ) \right ] (x,\xi ) = \frac {1}{2\pi}
\int_{\R} e^{-i\xi y} \varphi \left (x + \frac {1}{2} \h y  \right ) 
\overline {\phi \left (x - \frac {1}{2} \h y \right )}  dy. 
\ee
The map $\widetilde {\W^\hbar} (\varphi , \phi )$ is bilinear since
\be
\widetilde {\W^\hbar} (\lambda_1 \varphi_1 + \lambda_2 \varphi_2 , \phi ) = \lambda_1 \widetilde {\W^\hbar} (\varphi_1 , \phi )+ \lambda_2 \widetilde {\W^\hbar} ( \varphi_2 , \phi )
\ee
and 
\be
\widetilde {\W^\hbar} (\varphi , \mu_1 \phi_1 + \mu_2 \phi_2 ) = \overline {\mu_1} \widetilde {\W^\hbar} (\varphi , \phi_1 )+ \overline {\mu_2} \widetilde {\W^\hbar} ( \varphi , \phi_2 )
\ee
for any $\varphi, \, \varphi_1 ,\, \varphi_2 ,\, \phi ,\, \phi_1 ,\, \phi_2 \in L^2 (\R , dx)$ and any $\lambda_1 , \, \lambda_2 ,\, \mu_1 ,\, \mu_2 \in \C$.
Furthermore,
\bee
 \widetilde {\W^\hbar} (\varphi , \phi ) = \overline { \widetilde {\W^\hbar} (\phi , \varphi ) } \, . \label {Eq3}
 \eee
 
\begin {definition}\label {Def_1}
Let $\varphi \in L^2 (\R ,dx)$, we define the (semiclassical) Wigner transform of the {\bf pure state} $\varphi$ the function
\bee
W^\hbar (x,\xi ):= \left [  {\W^\hbar} \varphi \right ] (x,\xi ) = \frac {1}{2\pi}
\int_{\R} e^{-i\xi y} \varphi \left (x + \frac {1}{2} \h y  \right ) 
\overline {\varphi \left (x - \frac {1}{2} \h y \right )}  dy. \label {Eq4}
\eee
\end {definition}

We restrict our attention 
to the one-dimensional model where $d=1$; in the case of higher dimension $d>1$ then the numerical pre-factor $(2\pi)^{-1}$ must be replaced by $(2\pi)^{-d}$, 
$x$, $\xi$ and $y$ belongs to $\R^d$ and the product $\xi y$ must be replaced by the scalar product $\xi \cdot y$. 

Since $\W^\hbar \varphi = \widetilde {\W^\hbar} (\varphi , \varphi)$ then it is quite obvious to observe that the Wigner transform is {\bf not} linear; i.e. 
\be
\W^\h (\lambda \varphi)= \lambda^2 \W^\h (\varphi ) \ \mbox { and } \ \W^\h (\varphi_1 + \varphi_2 ) \not= \W^\h (\varphi_1) + \W^\h (\varphi_2 ) 
\ee
and the map $\W^\h$ {\bf is not a linear operator}. \ Thus, if we denote $\D = \W^\hbar \left [ L^2 (\R , dx) \right ]$ the image of $L^2 (\R , dx)$ via the Wigner 
transform it follows that the {\bf set} $\D$ {\bf is not a linear space} (see Remark \ref {Nota3}).

If $\hat \rho$ is a density matrix that represents a {\bf mixed quantum state} then the Wigner representation is given by
\bee
\frac {1}{2\pi} \int_{\R} e^{-i y \xi} \left \langle x + \frac 12 \hbar y \left | \hat \rho \right |  x - \frac 12 \hbar y \right \rangle dy \, . \label {Eq5}
\eee
In the case of pure quantum state then $\hat \rho = |\varphi \rangle \langle \varphi |$ and thus (\ref {Eq5}) takes the form (\ref {Eq4}). \ This paper mostly 
concerns the Wigner transform of pure states.

Here, we recall some basic facts about the Wigner transform:

\begin {itemize}

\item [i.] The Wigner transform is not an injective map because $\W^\h (\varphi ) = \W^\h \left ( e^{i\theta} \varphi \right ) $ for any real-valued parameter $\theta$ 
independent of $x$. \ In fact, in Theorem \ref {Teo_1} we'll prove that if $\W^\hbar (\varphi )= \W^\hbar (\phi )$ for some 
$\varphi ,\, \phi \in L^2 (\R , dx)$ then $\varphi = \lambda \phi$ for some terms $\lambda $ independent of $x$ and such that $|\lambda |=1$.

\item [ii.] $W^\h$ is a {\bf real-valued function}, indeed it directly comes from (\ref {Eq3}) and since $\W^\hbar (\varphi )= \widetilde 
{\W^\hbar} (\varphi , \varphi )$. \ We must remark that $W^\h (x,\xi)$ may be always a non-negative function (see, e.g., Example \ref {Es2}) or 
it may also take negative values (see, e.g., Example (\ref {Es1})). \ In fact, Hudson \cite  {H} proved that the Wigner transform $W^\hbar (x,\xi)$ of a 
wave-function $\varphi (x)\in L^2$ is always non-negative if, and only if, $\varphi (x)= e^{-\frac 12 (ax^2 + bx+c)}$, $\Re a>0$, is a Gaussian-type 
function (we should also mention that Janssen \cite {Janssen} extended such an analysis to non-$L^2$ functions). 

\item [iii.] If $\varphi (x)$ is an even/odd function, i.e. $\varphi (- x) = \pm \varphi (x)$ then $W^\h (-x,\xi ) = W^\h (x,-\xi )$.

\item [iv.] The dependence on $\hbar$ is as follows $
W^\h (x,\xi  )= \frac {1}{\h} W^1 \left ( x, \frac {\xi}{\h} \right )$.

\item [v.] For any $\alpha ,\, \beta >0$ it follows that
\be
\left [ {\W}^\h (\varphi ) \right ] (\alpha x,\xi ) = \alpha \left [ {\W}^\h \left ( T_\alpha \varphi \right ) \right ] (x,\alpha \xi )
\ee
and
\be
\left [ {\W}^\h (\varphi ) \right ] ( x,\beta \xi ) = \frac {1}{\beta} \left [ {\W}^\h \left ( T_{\beta^{-1}} \varphi \right ) \right ] (\beta x, \xi )
\ee
where $\left [ T_\alpha \varphi \right ] (x) = \varphi (\alpha x)$.

\item [vi.] If a function $u (x,\xi ; \h)$ is the Wigner transform of a function $\varphi $ depending on $\h$ of the form 
\be
\varphi (x; \h ) = \left [ \beta (\h ) T_{\alpha (\h )} \phi \right ] (x) = \beta (\h )\phi \left [ \alpha (\h ) x \right ]\, , 
\ee
for some function $\phi \in L^2 (\R , dx)$ independent of $\h$ and some functions $\alpha (\h )$ and $\beta (\h )$ such that $\alpha (\h ) \not= 0$ and $\beta (\h ) 
\not= 0$ for any $\h$, then $u (x,\xi ; \h )$ may be written as 
\be
u (x,\xi ; \h )&=& \left [ \W^\h \left ( \beta (\h ) T_{\alpha (\h )} \phi \right ) \right ] (x, \xi ) = \frac {\beta^2 (\h )}{\alpha (\h )} \left [ \W^\h (\phi ) \right ] 
\left ( \alpha (\h ) x, \frac {\xi } {\alpha (\h )} \right ) \\
&=& \frac {\beta^2 (\h )}{\h \alpha (\h )} w\left ( \alpha (\h )x, \frac {\xi}{\h \alpha (\h )} \right )
\ee
for some function $w (x,\xi )$ independent of $\h$.

\end {itemize}

Here, we collect some properties of the Wigner transform \cite {Barletti, Case, Folland, HOSW, Lee, Schleich}.

\begin {lemma} \label {Lem_1} Let $\varphi \in L^2 (\R , dx )$ and let $W^\h = {\W}^\h (\varphi )$, then 
\be
\int_{\R} W^\h (x,\xi ) d\xi = |\varphi (x) |^2 \ \mbox { and } \ \int_{\R} W^\h (x,\xi ) dx = 
\frac {2\pi}{ \h} \left |\hat \varphi \left ( \frac {\xi}{\h} \right ) \right |^2
\ee
where $\hat \varphi$ is the Fourier transform of $\varphi$.
\end {lemma}

\begin {proof}
We remark that 
\be
W^\h (x,\xi ) = \hat \nu_\h (\xi ; x) := \frac {1}{2\pi} \int_{\R} e^{-iy\xi } \nu_\h (y;x) dy 
\ee
where $\hat \nu_\h$ is the Fourier transform of
\bee
\nu_\h (y;x) := \varphi \left ( x + \frac 12 \h y  \right ) \overline { \varphi \left ( x - \frac 12 \h y \right )} \, . \label {Eq6}
\eee
Then it follows that
\be
\int_{\R} W^\h (x,\xi ) d\xi = \int_{\R} \hat \nu_\h (\xi ; x ) e^{i \xi \cdot 0} d \xi = \nu_\h (0;x) = |\varphi (x)|^2 \, . 
\ee
In order to prove the second statement let 
\be
\int_{\R} W^\h (x,\xi ) dx &=& \frac {1}{2\pi} \iint_{\R^2} e^{-iy\xi } \varphi \left ( x + \frac 12 \h y \right ) 
\overline {\varphi \left ( x - \frac 12 \h y \right )} dx dy \\ 
&=& \frac {1 }{2\pi \h} \iint_{\R^2} e^{-i(u-v)\xi/\h } \varphi ( u ) \overline {\varphi ( v )} du dv \\ 
&=& \frac {2\pi}{ \h} \hat \varphi \left ( \frac {\xi}{\h} \right ) \overline {\hat \varphi \left ( \frac {\xi}{\h} \right )}
\ee
where we set $u=x+ \frac 12 \h y$ and $v=x-\frac 12 \h y$. \end {proof}

\begin {remark} \label {Nota1}
From Lemma \ref {Lem_1} it follows that the expectation value $<a>$ of a classical observable $a=a(x,\xi)$ is given by
\be
<a> = \langle \psi , A \psi \rangle = \iint_{\R \times \R} a(x,\xi ) W^\hbar (x,\xi ) dx \, d\xi 
\ee
where $A$ is the linear operator associated to the classical observable and where we adopt the notation:
 \be
 \langle f, g \rangle := \langle f , g \rangle_{L^2} = \int_{\R} \overline {f(x)} g(x) dx\, , \ f,g \in L^2 (\R ,dx) \, . 
 \ee
 
\end {remark}

\begin {lemma} \label {Lem_2}
Let $\varphi \in L^2 (\R , dx )$ and let $W^\h = {\W}^\h (\varphi )$, then 
\be
\iint_{\R \times \R} W^\h (x,\xi ) d\xi dx = \| \varphi (\cdot )\|_{L^2}^2 \, .
\ee
\end {lemma}

\begin {proof}

From Lemma \ref {Lem_1} immediately follows that 
\be
\iint_{\R \times \R} W^\h (x,\xi ) d\xi dx = \int_{\R} |\varphi (x)|^2 dx = \| \varphi (\cdot ) \|^2_{L^2 (\R , dx )} \, . 
\ee
\end {proof}

\begin {remark} \label {Nota2}
The previous Lemma does not prove that $W^\h (x,\xi )$ belongs to $L^1 (\R \times \R , dx d\xi)$. \ Indeed, in Example \ref {Es1}, we consider a function 
$\varphi \in L^2 (\R , dx)$ such that its Wigner transform does not belongs to $L^1  (\R \times \R , dx d\xi)$.
\end {remark}

\begin {lemma} \label {Lem_3}
Let $\varphi_1 , \, \varphi_2 \in L^2 (\R , dx)$ and let $W^\h_j = {\W}^\h (\varphi_j )$, $j=1,2$, then 
\be
\left | \langle \varphi_1 , \varphi_2 \rangle_{L^2 (\R , dx)} \right |^2 = 2 \pi \h \langle W_1^\h , W_2^\h \rangle_{L^2 (\R \times \R , dx d\xi )} \, .  
\ee
In particular, if $W^\h = {\W}^\h (\varphi )$, $\varphi \in L^2 (\R ,dx)$, then $W^\h \in L^2 (\R \times \R , dx d\xi )$ and 
\be
\| W^\h (\cdot , \cdot ) \|_{L^2 (\R \times \R, dx d\xi)} = \sqrt {\frac {1}{2\pi\h}} 
\| \varphi (\cdot ) \|_{L^2 (\R , dx)}^2\, .
\ee
\end {lemma}

\begin {proof}
Recalling that 
\bee
\frac {1}{2\pi } \int_{\R} e^{iu\xi } d \xi = \delta (u)
\label {Eq7}
\eee
is the Dirac's $\delta$ distribution, then a straightforward calculation gives that 
\be
&& \langle W_1^\h , W_2^\h \rangle_{L^2 (\R \times \R, dx d\xi)} 
= \frac {1}{(2\pi )^2} \int_{\R \times \R} dx d\xi \int_{\R} dy_1 \int_{\R} d y_2 e^{i(y_1-y_2) \xi } \times \\ 
&& \ \ \times 
\varphi_1 \left ( x - \frac 12 \h y_1 \right ) \overline {\varphi_1 \left ( x + \frac 12 \h y_1 \right )}
\varphi_2 \left ( x + \frac 12 \h y_2 \right ) \overline {\varphi_2 \left ( x - \frac 12 \h y_2 \right )} \\
&& \ \ = \frac {1}{2\pi \h} \left | \langle \varphi_1 , \varphi_2 \rangle_{L^2 (\R , dx)} \right |^2
\ee
\end {proof}

\begin {remark} \label {Nota2Bis} Let $p\in [+1,+\infty ]$, and let 
\be
L^p_{\R} (\R \times \R , dx d\xi ) := \left \{ w \in L^p (\R \times \R , dx d\xi ) \ : \ w (x,\xi )\ \mbox {is real-valued} \right \} \, .
\ee
Then, by means of the previous Lemma we can conclude that the Wigner transform maps the Hilbert space $L^2 (\R , dx)$ in the Hilbert space $L^2_{\R} 
(\R \times \R , dx d\xi )$:
\be
\varphi (x) \in L^2 (\R , dx) \to \left [ \W^\h (\varphi )\right ] (x, \xi ) \in L^2_{\R} (\R \times \R , dx d\xi )\, . 
\ee
\end {remark}

In fact, $ \W^\h (\varphi ) \in L^2_{\R} \cap L^\infty_{\R}$; indeed:

\begin {lemma} \label {Lem_4_Bis}
Let $\varphi \in L^2 (\R , dx)$ and let $W^\h = {\W}^\h (\varphi )$, then 
\be
\| W^\h (x,\xi ) \|_{L^{\infty} (\R \times \R )} \le \frac {1}{\pi \h} \| \varphi (\cdot ) \|^2_{L^2 (\R ,dx)} \, .
\ee
\end {lemma}

\begin {proof}
The proof is quite simple: 
\be
\left | W^\h (x,\xi )\right | &\le & \frac {1}{2\pi} \int_{\R} \left | \varphi \left ( x+ \frac 12 \h y \right ) \right |\, 
\left | \varphi \left ( x- \frac 12 \h y \right ) \right | dy \\
&\le & \frac {1}{2\pi} \left \{ \int_{\R} \left | \varphi \left ( x+ \frac 12 \h y \right ) \right |^2 dy \right \}^{1/2} 
\left \{ \int_{\R} \left | \varphi \left ( x- \frac 12 \h y \right ) \right |^2 dy \right \}^{1/2} \\ 
&= & \frac {1}{\h \pi} \left \{ \int_{\R} \left | \varphi (u ) \right |^2 du \right \}^{1/2} 
\left \{ \int_{\R} \left | \varphi ( v ) \right |^2 dv \right \}^{1/2} = \frac {1}{\pi \h} \| \varphi (\cdot ) \|^2_{L^2 (\R ,dx)}
\ee
where we set $u=x+ \frac 12 \h y$ and $ v=x- \frac 12 \h y$.
\end {proof}

\subsection {Reversibility of the Wigner transform}\label {Sez2_2}

If $\varphi \in L^2$ is a {\bf positive real-valued function} then the transformation $\W^\h (\varphi )$ is invertible; indeed, from Lemma \ref {Lem_1} we can obtain 
the absolute value of the wave-function:
\be
\varphi (x) = |\varphi (x )|= \sqrt {\int_\R W^\h (x ,\xi ) d\xi }\, .
\ee

In general, since the wave-function $\varphi$ is not a real-valued function or it does not has a definite sign, the Wigner transform can be inverted as follows.

\begin {lemma} \label {Lem_4}
Let $\varphi \in L^2 (\R , dx)$ and let $W^\h = {\W}^\h (\varphi )$, and let $x^\star$ be such that $\int_\R W^\h (x^\star ,\xi ) d\xi \not= 0$. \ Then 
$\varphi (x^\star ) \not= 0$ and 
\bee
\varphi (x) = \frac {1}{\overline {\varphi (x^\star )}} \int_{\R} W^\h \left ( \frac {x+x^\star}{2} ,\xi \right ) e^{i(x-x^\star )\xi /\h } d\xi \, . \label {Eq8}
\eee
\end {lemma}

\begin {proof}
Since $\int_\R W^\h (x^\star ,\xi ) d\xi = |\varphi (x^\star )|^2$ then $\varphi (x^\star ) \not= 0$ immediately follows from the assumption. \ For any real-valued 
function $f(x)$ it follows that 
\be
\int_{\R} W^\h \left [ f(x) ,\xi \right ] e^{ix \xi } d \xi &=& \frac {1}{2\pi} \iint_{\R^2} d\xi dy 
e^{i(x-y)\xi } \varphi \left ( f(x) + \frac 12 \h y \right ) \overline {\varphi \left ( f(x) - \frac 12 \h y \right )} \\ 
&=& \varphi \left ( f(x) + \frac 12 \h x \right ) \overline {\varphi \left ( f(x) - \frac 12 \h x \right )}
\ee
since (\ref {Eq7}). \ In particular, if we set $f(x) =\frac 12 \h x + x^\star$ then, for any $x\in \R$, it follows that
\be
\int_{\R} W^\h \left ( \frac 12 \h x + x^\star ,\xi \right ) e^{ix \xi } d \xi = \varphi (\h x + x^\star ) \overline {\varphi (x^\star )}\, .
\ee
If we call $z=\h x+x^\star$, i.e. $x= \frac {z-x^\star}{\h}$, then we obtain 
\be
{\overline {\varphi (x^\star )}} \varphi (z) =  \int_{\R} W^\h \left ( \frac {z+x^\star}{2} ,\xi \right ) e^{i(z-x^\star )\xi /\h } d\xi 
\ee
from which (\ref {Eq8}) follows. \end {proof}

If we denote $\theta^\star $ the phase of $\varphi (x^\star )$ then (\ref {Eq8}) takes the form
\bee
\varphi (x) = \frac {e^{i\theta^\star }}{\sqrt {\left | \int_\R W^\h (x^\star ,\xi ) d\xi \right |}} \int_{\R} W^\h \left ( \frac {x+x^\star}{2} ,\xi \right ) 
e^{i(x-x^\star )\xi /\h } d\xi \, . \label {Eq9}
\eee
Thus, the Wigner transform $\W^\h$ is invertible up to a phase factor. \ That is we have proved that:

\begin {theorem} \label {Teo_1}
The map
\be
\W^\h :L^2 (\R ,dx) \to \D :=\W^\h \left ( L^2 (\R ,dx) \right ) \subseteq L^2_{\R} (\R \times \R , dx d\xi )
\ee
is, up to a phase factor independent of $x$, a  one-to-one map and the inverse map $\left [ \W^\h\right ]^{-1}$ is given by (\ref {Eq9}).
\end {theorem}

Actually $\D \subset L^2_\R$ and a question arises: which conditions on a real-valued square integrable function $w(x,\xi) \in L^2_\R (\R \times \R , dx d\xi )$ must 
be satisfied in order that it is a Wigner function (of a pure state), i.e. $w \in \D$? \ A first answer to this question has been given by \cite {HOSW}; they stated 
a set of conditions that are necessary and sufficient for a function $w\in L^2_\R (\R \times \R , dx d\xi )$ to be a Wigner transform (of a non necessarily pure state): 
that is $w (x,\xi)$ must be normalized, in the sense that $\iint_{\R^2} w(x,\xi ) dx d\xi =1$, and furthermore 
\bee
\langle w , \W^\hbar (\phi ) \rangle_{L^2 (\R \times \R , dx d\xi )} \ge 0 \, , \ \forall \phi \in L^2 (\R , dx) \, . \label {Eq10}
\eee
As pointed out by Narcowich and O'Connel \cite {NOC} such a set of conditions is not very satisfactory because (\ref {Eq10}) is quite hard to check from a practical 
point of view; indeed, they proposed a second set of conditions. \ Namely, let $\tilde w(u,v)$ be the symplectic Fourier transform of $w(x,\xi )$; then 
$w \in L^2_\R (\R \times \R , dx d\xi )$ is a Wigner transform function if, and only if, $\tilde w (0,0)=1$ and $\tilde w (u,v)$ is continuous and it is of 
$\hbar$-positive type (see \cite {K,LMS1,LMS2} for a definition of $\hbar$-positive type). \ Unfortunately, such a criterion cannot distinguish between Wigner 
functions associated with pure states or mixed states. \  Tatarski \cite {Tat} introduced a necessary and sufficient condition for a function 
$w \in L^2_{\R} (\R \times \R , dx d\xi)$ to describe a pure quantum state; in particular, he proved that $w\in {\mathcal D}$, that is $w=\W^\hbar (\varphi)$ 
for some $\varphi \in L^2$, if, and only if, 
\bee
\frac {\partial^2}{\partial x_1 \partial x_2} \ln \left [ Q(x_1,x_2 ) \right ] =0 \label {Eq11}
\eee
where 
\be
Q(x_1 ,x_2 ) := \int_{\R} e^{i \xi (x_1 -x_2 )/\hbar} w\left ( \frac {x_1-x_2}{2} , \xi \right ) d \xi \, .
\ee
If (\ref {Eq11}) is satisfied the wave-function $\varphi (x)$ can be recovered from $w(x,\xi )$ by (\ref {Eq9}), up to a phase factor independent of $x$. 

Finally, we should also mention the results by \cite {DP} where the authors make use of the notion of Narkowich-Wigner spectrum in order to characterize the 
Wigner functions of a pure state.

\subsection {Continuity of the Wigner transform} \label {Sez2_3}

We have seen that the Wigner transform maps the space $L^2 (\R , dx)$ in the space $L^2_\R (\R \times \R , dx d\xi )$. \ We prove now that such a map is continuous; in 
particular the following estimates hold true.

\begin {theorem} \label {Teo_2}
  Let $w_j =\W^\h (\varphi_j )$, where $\varphi_j \in L^2$ and $j=1,2$, then
  \be
  \| w_1 - w_2 \|_{L^2 (\R \times \R, dx d\xi )} \le C \inf_{\theta \in \R}  \| e^{i\theta} \varphi_1 - \varphi_2 \|_{L^2 (\R , dx)}
  \ee
where
\be
C = \frac {\sqrt {2} }{\sqrt {\hbar}} \left (\| \varphi_1 \|_{L^2 (\R , dx)} + \| \varphi_2 \|_{L^2 (\R , dx)} \right ) \, . 
\ee
\end {theorem}

\begin {proof} Let us denote
\be
\varphi_{j,\pm} := \varphi_{j,\pm} (x,y ) = e^{i\theta } \varphi_j \left ( x \pm \frac 12 \h y \right )  \, , \ j=1,2 \, , \ 
\ee
where $\theta \in \R$ is independent of $x$, and 
\be 
\phi_\pm := \phi_\pm (x,y)= \varphi_{1,\pm} (x,y)- \varphi_{2,\pm}(x,y)\, , 
\ee
then
\bee
 &&\left [ w_1 (x,\xi ) - w_2 (x,\xi ) \right ] = \frac {1}{2\pi} \int_\R e^{-i \xi y} \left [ \varphi_{1,+} \overline {\varphi_{1,-}} - \varphi_{2,+} 
 \overline {\varphi_{2,-}} \right ] dy \label {Eq12}\\ 
&& \ \ =  \frac {1}{2\pi} \int_\R e^{-i \xi y} \left [ (\varphi_{1,+} - \varphi_{2,+} ) \overline {\varphi_{1,-}} + \varphi_{2,+} ( \overline {\varphi_{1,-}}- 
\overline {\varphi_{2,-}}) \right ] dy \label {Eq13}\\ 
&& \ \ = \left [ {\mathcal F}_y (\phi_+ \overline {\varphi_{1,-}} )\right ] (x,\xi ) + \left [ {\mathcal F}_y (\overline {\phi_-}  \varphi_{2,+} )\right ] (x,\xi ) 
\label {Eq14}
\eee
where $ {\mathcal F}_y (f)$ denotes the Fourier transform (with respect to $y$) of a function $f$. 

Furthermore, recalling that $\| {\mathcal F}_y f \|_{L^2} = \|  f \|_{L^2}$ then 
\be
&& \| w_1 (x,\xi ) - w_2 (x,\xi )\|_{L^2 (\R \times \R , dx d\xi )}^2 = \int_\R dx \int_\R d \xi |w_1 (x,\xi ) - w_2 (x,\xi )|^2  \\ 
&& \ = \int_\R d x \| w_1 (x, \cdot ) - w_2 (x, \cdot ) \|^2_{L^2 (\R , d\xi )} \\
&& \ = \int_\R d x \| \left [ {\mathcal F}_y (\phi_+ \overline {\varphi_{1,-}} )\right ] (x,\cdot  ) + \left [ {\mathcal F}_y (\overline {\phi_-}  
\varphi_{2,+} )\right ] (x, \cdot  ) \|^2_{L^2 (\R , d\xi )} \\ 
&& \ = \int_\R d x \|  (\phi_+ \overline {\varphi_{1,-}} )(x,\cdot  ) + (\overline {\phi_-}  \varphi_{2,+} )(x, \cdot  ) \|^2_{L^2 (\R , d\xi )} \\ 
&& \ \le 2 \int_\R d x \| (\phi_+ \overline {\varphi_{1,-} })(x,\cdot  )  \|^2_{L^2 (\R , d\xi )} + 
2 \int_\R d x \| (\phi_- \overline {\varphi_{2,+}} ) (x,
\cdot  ) \|^2_{L^2 (\R , d\xi )} \, . 
\ee
We restrict now our attention to the first integral 
\be
&& 
\int_\R d x \| (\phi_+ \overline {\varphi_{1,-}} ) (x,\cdot  )  \|^2_{L^2 (\R , dy )} = \int_\R d x \int_{\R} dy | (\phi_+ \overline {\varphi_{1,-}} ) (x,y  )  |^2 \\ 
&& \ \ = \int_\R d x \int_{\R} dy \left | e^{i\theta} \varphi_1 \left ( x + \frac 12 \hbar y \right ) - \varphi_2 \left ( x + \frac 12 \hbar y \right ) \right |^2 
\left | \overline {e^{i\theta} \varphi_1 \left ( x - \frac 12 \hbar y \right ) } \right |^2 
\ee
By means of the change of variable $x\to z = x + \frac 12 \hbar y $ the integral above becomes 
\be
\int_\R d z \left | e^{i\theta} \varphi_1 \left ( z \right ) - \varphi_2 \left ( z \right ) \right |^2 \int_{\R} dy  \left | \overline { e^{i\theta } \varphi_1 
\left ( z -  \hbar y \right ) }\right |^2 =\frac {1}{\hbar} \| e^{i\theta } \varphi_1 - \varphi_2 \|^2_{L^2} \| \varphi_1 \|_{L^2}^2 \, . 
\ee
The other integral can be similarly treated proving thus the Theorem. \end {proof}

We can also prove that 

\begin {lemma} \label {Lem_6}
We have that 
\bee
\| w_1 - w_2 \|_{L^\infty (\R \times \R ,dx d\xi ) } \le C \inf_{\theta \in \R}  \| e^{i\theta } \varphi_1 - \varphi_2 \|_{L^2 (\R , dx)}\label {Eq15}
\eee
where $C= \frac {1}{\pi \h } \left [ \| \varphi_1 \|_{L^2 (\R , dx)} + \| \varphi_2 \|_{L^2 (\R , dx)} \right ]$.
\end {lemma}

\begin {proof}
Estimate (\ref {Eq15}) immediately follows from (\ref {Eq13}); indeed:
\be
\left | w_1 (x,\xi ) - w_2 (x,\xi )  \right | &\le& 
\frac {1}{2\pi} \int_\R \left [ \left |\varphi_{1,+} - \varphi_{2,+} \right | \, \left | \bar \varphi_{1,-} \right |+ \left | \varphi_{2,+} \right | \, \left |\overline 
{\varphi_{1,-}} - \overline { \varphi_{2,-}} \right | \right ] dy \\
&\le & \frac {1}{2\pi} \left [  \left \|\varphi_{1,+} (x,\cdot )- \varphi_{2,+} (x,\cdot )\right \|_{L^2 (\R , dy)} \left \| \overline {\varphi_{1,-}  (x,\cdot )}
\right \|_{L^2 (\R , dy)} + \right. \\ 
&& \left. \ \ + \left \|\varphi_{1,-} (x,\cdot )- \varphi_{2,-} (x,\cdot )\right \|_{L^2 (\R , dy)} \left \| \varphi_{2,+}  (x,\cdot )\right \|_{L^2 (\R , dy)}
\right ]
\ee
where, for instance, 
\be
\left \| \varphi_{2,+}  (x,\cdot )\right \|_{L^2 (\R , dy)} &=& \left [ \int_\R |\varphi_{2,+} (x,y) |^2 dy \right ]^{1/2} = 
\left [ \int_\R \left |\varphi_{2} \left (x+ \frac 12 \h y \right ) \right |^2 dy \right ]^{1/2} \\
&=& \left [ \frac {2}{\h} \int_\R |\varphi_{2} (z) |^2 dz \right ]^{1/2} = \sqrt {\frac {2}{\h} } \left \| \varphi_{2} \right \|_{L^2 (\R , dx)} \, . 
\ee
\end {proof}

\subsection {Examples} \label {Sez2_4} Let us consider some examples of computation of Wigner transform of normalized functions $\varphi \in L^2 (\R , dx)$.

\subsubsection {Example 1.} \label {Es1} \ Let $\varphi (x) = \frac {1}{\sqrt {2R}} \chi_{[-R,+R]} (x)$ where $\chi$ is the characteristic function, 
i.e. $\chi_A (x)=1$ if $x \in A$ and $\chi_A(x)=0$ if $x\notin A$, and $R>0$ is fixed. \ Then
\be
W^h (x,\xi ) = \frac {1}{4\pi R} \int_\R e^{-i \xi y} \chi_{[-R,+R]} \left ( x + \frac 12 \h y \right ) \chi_{[-R,+R]} \left ( x - \frac 12 \h y \right ) dy 
\ee
is an even function with respect to $x$. \ We assume, for argument's sake, that $x \ge 0$. \ Hence
\be
W^\h (x,\xi )= 0 \ \mbox { if } \ x \ge R 
\ee
and 
\be
W^\h (x,\xi )= \frac {1}{4\pi R} \int_{-2(R-x)/\h}^{2(R-x)/\h} e^{-i \xi y} dy = \frac {1}{2\pi R \xi} \sin \left [ \frac {2\xi}{\h} (R-x) \right ] 
\ \mbox { if } \ x \le R \, . 
\ee
In conclusion
\be
W^\h (x,\xi ) = \frac {1}{2\pi R \xi} \sin \left [ \frac {2\xi}{\h} (R-|x|) \right ] \chi_{[-R,+R]}(x) \, . 
\ee
One should remark (see Appendix \ref {App0}) that %
\bee
W^\hbar \notin L^1 (\R \times \R , dx d\xi )\, . \label {Eq16}
\eee

\subsubsection {Example 2.} \label {Es2} \ Let $\varphi (x) := \frac {1}{\sqrt[4]{\pi}} e^{-x^2/2}$; then a straightforward calculation gives that 
\be
W^\h (x,\xi )= \frac {1}{\pi \h} e^{- \left (x^2 + \frac {\xi^2}{\h^2}\right )}  \, . 
\ee
Now, let $\varphi_\alpha (x)= \sqrt [4]{\frac {\alpha^2}{\pi}} e^{-\alpha^2 x^2 /2} = \sqrt {\alpha} \varphi (\alpha x)$. \ Then  
\be
W^\h_\alpha (x,\xi ) &=& [\W^\h  \left ( \varphi_\alpha \right )] (x,\xi ) =  \left [\W^\h \left ( \sqrt {\alpha } T_\alpha \varphi \right ) \right ] (x,\xi )= 
\alpha \left [\W^\h \left ( T_\alpha \varphi \right ) \right ] (x,\xi )\\ &=& 
\left [\W^\h (\varphi ) \right ] \left (\alpha x,\frac {\xi}{\alpha} \right )= 
\frac {1}{\pi \h} e^{-\alpha^2 x^2 - \frac {\xi^2}{\alpha^2 \h^2} }\, .
\ee
In particular, if $\alpha =1/\sqrt {\hbar}$ then 
\be
\varphi (x) =  \frac {1}{\sqrt[4]{\hbar \pi}} e^{-x^2/2\hbar}
\ee
has Wigner transform 
\be
W^\hbar (x,\xi )= \frac {1}{\pi \hbar} e^{-\frac {x^2 + \xi^2}{ \hbar}}\, .
\ee

\subsubsection {Example 3.} \label {Es3} \ Let 
\be
\varphi (x) =\exp \left [ - \frac 12 \left ((a_1 +i a_2) x^2 + (b_1+ib_2)x + (c_1+ic_2) \right ) \right ]
\ee
be the Gaussian-type function considered by Hudson \cite {H}, where $a_1, a_2,b_1 , b_2 , c_1 , c_2 \in \R$ and 
$a_1 >0$. \ By means of a straightforward calculation it turns out that $W^\h (x,\xi ) = \left [ \W^\h (\varphi )\right ] (x,\xi )$ is still a real-valued and positive  Gaussian-type 
function given by
\be
W^\h (x, \xi )= \frac {\exp \left [ - \frac { 
4\h^2 (a_1^2+a_2^2) x^2 + 8 a_2 \h x \xi + 4 \h^2 (a_1 b_1 + a_2 b_2)x + b_2^2 \h^2 + 4 a_1 c_1 \h^2 + 4 b_2 \h \xi + 4 \xi^2 
}{4a_1 \h^2}  \right ]}{\h \sqrt {\pi a_1}} 
\ee

\subsubsection {Example 4.} \label {Es4} \ Let $\varphi_n (x) =  H_n (x)e^{-x^2/2}$ where $H_n (x)$ is the $n$-th Hermite polynomial; then it is not hard to see \cite {Folland} that its Wigner transform (for $\h =1$) is given by 
\be
W^1_n (x,\eta )&=& \left [ \W^1 (\varphi_n) \right ] (x,\xi ) = \frac {(-1)^n}{\pi} e^{-\left ( x^2 + \eta^2 \right )}L_n \left [ 2 (x^2 + \eta^2 ) \right ]
\ee
where $L_n$ is the $n$-th Laguerre polynomial. \ Then 
\be
W^\h_n (x,\xi ) = \frac {(-1)^n}{\pi \h} e^{-\left ( x^2 + \frac {\xi^2}{\h^2} \right ) }L_n \left [ 2 x^2 + \frac {2\xi^2}{\h^2}  \right ] 
\ee

\begin {remark} \label {Nota3}
We recall that $L_0 (z)= 1$ and $L_1 (z)=1-x$; thus the function defined as 
\be
w(x,\xi ) &=& C_0 W^\h_0 (x,\xi )+ C_1 W^\h_1 (x,\xi) \\
&=& \frac {1}{\h \pi}e^{- \left ( x^2+\frac {\xi^2}{\h^2} \right )} \left [C_0 - C_1 \left ( 1- \left ( 2 {x^2} + \frac {2\xi^2}{\h^2} \right )  
\right ) \right ] \\
&=& \frac {C_0}{2\h \pi}e^{- \left ( x^2+\frac {\xi^2}{\h^2} \right )} \left [1 + 2{x^2} + \frac {2\xi^2}{\h^2} \right ]\, , 
\ee
where $C_1 =  \frac 12 C_0 >0$, is such that $w(x,\xi )>0$ everywhere. \ If $w \in \D$ then, by the Hudson's argument, $w(x,\xi )$ must be 
of the form discussed in the Example \ref {Es3}. \ Because this is not the case then $w\notin \D$ and thus we can conclude that the set $\D$ is not closed to 
linear combinations.
\end {remark}

\subsubsection {Example 5.} \label {Es5} \  {\bf Wavefunction of the free Schr\"odinger equation.} \ Let us consider the free linear Schr\"odinger equation 
$i \hbar\frac {\partial \psi}{\partial t}  = - {\hbar^2} \frac {\partial^2 \psi}{\partial x^2} $ with normalized initial condition 
\be
\psi_0 (x;\hbar ) = \sqrt [4]{\frac {2}{\pi}} e^{-x^2}\, ;
\ee
then it is well known that \cite {T}
\be
\psi (x,t; \hbar ) &=& \left [ U(t) \psi_0 \right ] (x;\hbar ) = \int_{\R} K_0 (x-y,t) \psi_0 (y;\hbar ) d y \\ 
&=& \sqrt [4] {\frac {2}{\pi}} \sqrt {\frac {i}{i-4\hbar t}} e^{-i \frac {x^2}{i-4\hbar t}} 
= \sqrt [4] {\frac {2}{\pi}} \sqrt {\frac {1-4i\hbar t }{16 \hbar^2 t^2 +1}} e^{-\frac {x^2(1-4i\hbar t )}{1+16\hbar^2 t^2}}
\ee
where the kernel $K_0$ is given by
\be
K_0 (z,t)= \sqrt {\frac {1}{4\pi i \hbar t}}e^{i \frac {z^2}{4\hbar t}} \, . 
\ee
A straightforward calculus gives that the Wigner transform $W^\hbar = \W^\hbar (\psi )$ of the wave-function $\psi (x,t;\hbar )$ is given by
\be
W^\hbar (x,\xi ,t )
= \frac {1}{\pi \hbar} e^{ -\frac {\xi^2}{2\hbar^2}} e^{-2(x-2\xi t)^2}\, . 
\ee

\subsubsection {Example 6.} \label {Es6} \ {\bf Stationary solution for a singular potential.} \ Let us consider the equation  
\bee
i \hbar \frac {\partial \psi}{\partial t}  = - \hbar^2\frac {\partial^2 \psi}{\partial x^2}   + \gamma \delta \psi \label {Eq17}
\eee 
where $\delta$ is the Dirac's delta distribution supported at $x=0$. \ We look for a stationary solution of the form 
$ \psi (x,t;\hbar )=e^{-iE t/\hbar } \varphi (x;\hbar )$ for some $E \in \R$. \ It is well known \cite {AGHH} that exactly one stationary solution 
occurs only if $\gamma <0$ and it is given by
\be
E = - \frac {\gamma^2}{4 \hbar^2} \ \mbox { and } \ \varphi (x)= \sqrt {\kappa}e^{-\kappa |x|}\, , \ \kappa = \frac {\sqrt {|E|}}{\hbar} = \frac {|\gamma|}{2\hbar^2}
\ee
By means of a straightforward calculation it turns out that the Wigner's transform is an even function with respect $x$ and $\xi$ both, 
$W^\h (-x , \xi )=W^\h (x,-\xi )=W^\h (x,\xi)$, given by
\bee
W^\h (x,\xi )
= \frac {\h \kappa^2}{\xi \pi({\kappa^2 \h^2 + 4 \xi^2})} e^{-2 \kappa |x|} \left [ \cos \left ( \frac {2|x|\xi}{\h}\right ) \xi + \kappa \h \sin 
\left ( \frac {2|x|\xi}{\h}\right ) \right ]
\, . \label {Eq18}
\eee

\subsubsection {Example 7.} \label {Es7} \ {\bf Stationary soliton.} \ Let us consider the Gross-Pitaevskii equation  
\be
i \hbar\frac {\partial \psi}{\partial t}   = - \hbar^2\frac {\partial^2 \psi}{\partial x^2}   + \nu |\psi |^2 \psi \, , \ \nu \in \R \, , 
\ee 
where we look for a stationary solution of the form $ \psi (x,t;\hbar )=e^{-iE t/\hbar } \varphi (x;\hbar )$ for some $E \in \R$ and a real-valued function 
$\varphi$. \ Then $\varphi (x;\hbar )$ is a normalized 
solution to the equation $- \hbar^2 \varphi'' + \nu \varphi^3 = E \varphi$, that is it is given by $\varphi (x)=A\mbox {sech} (Bx)$ where \cite {Davis} 
\be
A= \frac {\sqrt {-\nu }}{\sqrt {8}\hbar} \, , \ B= - \frac {\nu }{4\hbar^2} \, \ \mbox { and } \ E = - \frac {\nu^2}{16\hbar^2 }
\ee
provided that $\nu <0$. \ In conclusion:
\be
\psi (x,t;\hbar )= \frac {\sqrt {-\nu }}{\sqrt {8} \hbar} \mbox {sech} \left ( {\frac {-\nu}{4\hbar^2}} x \right )e^{i \nu^2 t/16 \hbar^3} \, . 
\ee
Its Wigner transform does not depend on time and it is given by 
\bee 
W^\hbar (x,\xi ;\hbar ) 
= \frac {1}{\hbar} \frac {\sin \left (\frac {2 x \xi }{\hbar} \right )} {\sinh \left ( \frac {\nu x}{2\hbar^2} \right ) 
\sinh \left (  \frac {4 \pi \xi \hbar}{\nu} \right )} \label {Eq19}
\eee
by means of a straightforward calculation \cite {KL} (see Appendix \ref {App1} for details).

\section {Schr\"odinger equation in the Wigner representation} \label {Sez3}

Here we deal with the one-dimensional semiclassical linear Schr\"odinger equation 
\bee
\left \{
\begin {array}{l}
i \hbar \frac {\partial \psi}{\partial t} = H\psi   \\ 
\psi (x,0;\hbar )=\psi_0 (x;\hbar ) 
\end {array}
\right. \ \mbox { where } \ H := - \frac {\hbar^2}{2m} \frac {\partial^2 }{\partial x^2} + V(x,t) \label {Eq20}
\eee
where $\hbar \ll 1$ is a semiclassical parameter; for argument's sake we choose the units such that $2m=1$. \ Equation (\ref {Eq20}) is a one-dimensional 
linear Schr\"odinger equation and, under some suitable assumptions on the real-valued potential $V(x,t)$, it has a global solution $\psi (x,t;\hbar )$ and the conservation 
of the norm $\| \psi (\cdot , t ;\hbar ) \|_{L^2} = \| \psi_0 (\cdot ;\hbar ) \|_{L^2}$ holds true.

In order to consider a different approach we make use of the (semiclassical) Wigner 
transform $W^\hbar (x,\xi , t ;\hbar )$ of $\psi (x,t;\hbar )$ defined as:
\be
W^\hbar (x,\xi ,t ;\hbar )= \frac {1}{2\pi} \int_{\R} e^{-iy \xi} \psi \left ( x + \frac 12 \hbar y , t ;\hbar \right ) 
\overline {\psi  \left ( x - \frac 12 \hbar y , t ;\hbar \right )} dy \, .
\ee
Hereafter, we simply denote
\be
W^\hbar (x,\xi , t) := W^\hbar (x,\xi , t;\hbar )\, . 
\ee

In the Wigner's representation the Schr\"odinger equation (\ref {Eq20}) takes the following form.

\begin {theorem} [Time-dependent Schr\"odinger equation in the Wigner representation] \label {Teo_3} 
Assume that $V(x,t)$ is a smooth real-valued function. \ The Wigner transform $W^\hbar (x,\xi , t)$ satisfies to the following formal equation
\bee
\left \{
\begin {array}{l}
\frac {\partial W^\hbar}{\partial t} = -2 \xi \frac {\partial W^\hbar}{\partial x} + \frac {\partial V}{\partial x} \frac {\partial W^\hbar}{\partial \xi} + 
\sum_{m =1}^\infty \frac {(-1)^m\hbar^{2m} }{2^{2m} (2m+1) !} \frac {\partial^{2m+1} V}{\partial x^{2m+1}} \frac {\partial^{2m+1}W^\hbar }{\partial \xi^{2m+1}}  \\
W^\hbar (x,\xi , 0)= W^\hbar_0 (x,\xi ) = \left [ \W^\hbar (\psi_0 ) \right ] (x, \xi )
\end {array}
\right. \, . 
\label {Eq21}
\eee
\end {theorem} 

\begin {remark} \label {Nota4}
In the semiclassical limit where $\hbar \to 0$ the dominant term of equation (\ref {Eq21}) is given by
\bee
\frac {\partial W^\hbar}{\partial t} = -2 \xi \frac {\partial W^\hbar}{\partial x} + \frac {\partial V}{\partial x} \frac {\partial W^\hbar}{\partial \xi}\label {Eq22}
\eee
and it coincides with the classical Liouville equation. \ If the potential $V(x)$ {\bf is independent of $t$} let $h (q,p) = p^2 + V(q)$ be the classical Hamiltonian 
(where ${2m}=1$) and let 
\bee
\left \{ 
\begin {array}{lcl}
\dot p &=& \frac {\partial h}{\partial q}\\
\dot q &=&- \frac {\partial h}{\partial p}
\end {array}
\right. \ \mbox { with initial condition } 
\left \{
\begin {array}{lcl}
 p(0)&=& \xi \\
 q(0) &=& x
\end {array}
\right. \, . \label {Eq23}
\eee
Then, the solution to (\ref {Eq22}) is given by
\be
W(x,\xi ,t) =W_0^\hbar \left [ q(t),p(t) \right ] = W_0^\hbar \left [ S^t (x,\xi ) \right ]
\ee
where $\left ( q(t),p(t) \right )= S^t (x,\xi )$ is the Hamiltonian flux associated to (\ref {Eq23}).
\end {remark}

\begin {theorem} [Time-independent Schr\"odinger equation in the Wigner representation] \label {Teo_4} Let 
$\psi (x,t;\hbar )= \varphi (x;\hbar )e^{-iEt/\hbar}$, where $\varphi \in L^2 (\R , dx)$, be a stationary solution to the time-dependent Schr\"odinger equation
\be
E \varphi = H \varphi \, , \ \mbox { where } \ H := - \frac {\hbar^2}{2m} \frac {\partial^2 }{\partial x^2} + V(x)\, , 
\ee
and where $V(x)$ is a smooth real-valued function independent of $t$. \ Then the Wigner transform $W^\hbar (x,\xi ) = \left [ \W^\hbar ( \varphi ) \right ] (x,\xi )$ is 
independent of $t$ and it satisfies to the following formal equation
\bee
E W^\hbar = -  \frac {\hbar^2}{4}  \frac {\partial^2 W^\hbar}{\partial x^2}  + \left [ \xi^2 + V \right ] W^\hbar  + 
\sum_{m =1}^\infty \frac {(-1)^m \hbar^{2m}}
{2^{2m}(2m) !}  \frac {d^{2m} V}{d x^{2m}}  \frac {\partial^{2m}W^\hbar }{\partial \xi^{2m}} \label {Eq24}
\eee
under the constrain
\bee
0 = -2 \xi \frac {\partial W^\hbar}{\partial x} + \frac {d V}{d x} \frac {\partial W^\hbar}{\partial \xi} + 
\sum_{m =1}^\infty \frac {(-1)^m\hbar^{2m} }{2^{2m} (2m+1) !} \frac {d^{2m+1} V}{dx^{2m+1}} \frac {\partial^{2m+1}W^\hbar }{\partial \xi^{2m+1}}  \, . 
\label {Eq25}
\eee
\end {theorem}

\begin {remark} \label {Nota5}
We may remark that the systems of equations (\ref {Eq21}) and (\ref {Eq24}-\ref  {Eq25}) are only formal because we don't make any statement about the 
convergence of the series. \ In fact, when the potential $V$ is a polynomial with finite degree with respect to $x$ then this problem does not occur. 
\end {remark}

\begin {remark} \label {Nota6}
One may, in principle, solve equations (\ref {Eq21}), or (\ref  {Eq24}-\ref {Eq25}),  by means of 
some approximation or numerical methods. \ This approach presents some problems: assume to find an approximate solution $w$ to 
(\ref {Eq21}) such that $\| W^\hbar - w \| \ll 1 $ with respect some norm, where $W^\hbar$ is the exact solution to (\ref {Eq21}). \ Some questions occur: 

\begin {itemize}

\item [{\it i}.] does $w \in \D$, that is the approximate solution $w$ is the Wigner transform of a pure state? \ This point has been discussed in Section \ref {Sez2_2}; 
but the criterion proposed by Narcowich and O'Connel, or by Tatarski, are quite hard to be implemented in explicit models.

\item [{\it ii}.] and in the affirmative case, what can we say about the difference $\| \psi - \phi \|_{L^2}$, where $\psi = \left [ \W^\hbar \right ]^{-1} ( W^\hbar )$ 
and $\phi = \left [ \W^\hbar \right ]^{-1} (w)$? \ In fact, we have proved in Theorem \ref {Teo_2} the continuity of the 
Wigner transform; but the continuity of the inverse is, as far as we know, an open problem: i.e., if $\varphi_i = \left [ {\mathcal W}^\hbar \right ]^{-1} (w_i) $ 
where $w_i \in \D$, $i=1,2$, we would require that
 \be
 \inf_{\theta } \| e^{i\theta }\varphi_1 - \varphi_2 \|_{L^p} \le C \| w_1 -w_2 \|_{L^q}^r
 \ee
 for some $p, q$ and $r$, and some positive constant $C$ depending on $\hbar$. \ However, some results hold true; for instance, let $a(x)$ be a bounded 
 classical observable, let $ <a>^j = \langle \varphi_j (\cdot ) ,a(\cdot ) \varphi_j (\cdot ) \rangle$, $j=1,2$, be the expectation value of the classical 
 observable $a(x)$ on the quantum state described by the wave-function $\varphi_j$. \ Then 
 \be
 | <a>^1 - <a>^2 | &=& \left | \int_{\R} a(x) \left [ | \varphi_1 (x) |^2 - | \varphi_2 (x) |^2 \right ] dx \right | \\ 
 &=& \left | \iint_{\R^2} a(x) \left [ w_1 (x,\xi) - w_2 (x,\xi ) \right ] d\xi dx \right | \\
 &\le & \| a\|_{L^\infty (\R , dx)} \| w_1 - w_2 \|_{L^1 (\R \times \R , dx d\xi )}
\ee 

\end {itemize}

\end {remark}

\subsection {Proofs of Theorems \ref {Teo_3} and \ref {Teo_4}} \label {Sez3_1}

Let us denote 
\bee
\psi_\pm := \psi \left ( x \pm \frac 12 \hbar y,t ; \hbar \right ) \ \mbox { and } \ V_\pm = V \left ( x \pm \frac 12 \hbar y  , t \right ).\label {Eq26} 
\eee
Then, from the Schr\"odinger equation (\ref {Eq20}) it follows that
\bee
i \hbar \dot \psi_+ = - \hbar^2 \psi_+ '' + V_+ \psi_+  \label {Eq27}
\eee
and 
\bee
-i \hbar \overline {\dot \psi_-} = - \hbar^2 \overline {\psi_- ''} + V_- \overline {\psi_-}  \label {Eq28}
\eee
since the potential $V$ is assumed to be a real-valued function. \ From these two equations it follows that
\bee
i \hbar \int_{\R} \dot \psi_+ \overline {\psi_-} e^{-i y \xi } dy = - \hbar^2  \int_{\R} \psi_+'' \overline  {\psi_-} e^{-i y \xi } dy +  
\int_{\R} \psi_+ \overline {\psi_-} V_+ e^{-i y \xi } dy  \label {Eq29}
\eee
and 
\bee
- i \hbar \int_{\R}  \psi_+ \overline {\dot \psi_-} e^{-i y \xi } dy = - \hbar^2  \int_{\R} \psi_+ \overline {\psi_-''} e^{-i y \xi } dy +  
\int_{\R} \psi_+ \overline { \psi_-} V_- e^{-i y \xi } dy  \label {Eq30}
\eee
If we take the difference between (\ref {Eq29}) and (\ref {Eq30}), and then the sum, it follows that
\bee
2\pi i \hbar \frac {\partial W^\hbar}{\partial t} = i \hbar \int_{\R} \left ( \dot \psi_+ \overline {\psi_-} + \psi_+ \overline {\dot \psi_-} \right ) e^{-i y \xi } dy
= - \hbar^2 F_1 + F_2  \label {Eq31}
\eee
and 
\bee
 i \hbar \int_{\R} \left ( \dot \psi_+ \overline {\psi_-} - \psi_+ \overline {\dot \psi_-} \right ) e^{-i y \xi } dy = - \hbar^2 F_3 + F_4  \label {Eq32}
\eee
where
\be
F_1 &=&   \int_{\R} \left ( \psi_+'' \overline {\psi_-} - \psi_+ \overline {\psi_-''} \right ) e^{-i y \xi } dy \\
F_2 &=&  \int_{\R} \left ( V_+ - V_- \right ) \psi_+ \overline {\psi_-} e^{-i y \xi } dy \\
F_3 &=&   \int_{\R} \left ( \psi_+'' \overline {\psi_-} + \psi_+ \overline {\psi_-''} \right ) e^{-i y \xi } dy \\
F_4 &=&  \int_{\R} \left ( V_+ + V_- \right ) \psi_+ \overline {\psi_-} e^{-i y \xi } dy \\
\ee

By means of straightforward calculations on can check that 
\bee
F_1 = \frac {4\pi i \xi}{\hbar} \frac {\partial W^\hbar}{\partial x}  \label {Eq33}
\eee
and that 
\bee
F_3 =  \pi \frac {\partial^2 W^\hbar}{\partial x^2}  - \frac {4\pi }{\hbar^2}\xi^2 W^\hbar \, .  \label {Eq34}
\eee

Concerning the other two terms $F_2 $ and $F_4$ we have that 
\bee
F_2 &=& 4\pi \sum_{m =0}^\infty \frac {1}{(2m+1) !} \frac {\partial^{2m+1} V}{\partial x^{2m+1}} \left ( 
\frac i2 \hbar \right )^{2m+1} \frac {\partial^{2m+1}W^\hbar }{\partial \xi^{2m+1}} \label {Eq35} \\ 
F_4 &=& 4\pi \sum_{m =0}^\infty \frac {1}{(2m) !} \frac {\partial^{2m} V}{\partial x^{2m}} \left ( 
\frac i2 \hbar \right )^{2m} \frac {\partial^{2m}W^\hbar }{\partial \xi^{2m}} \label {Eq36}
\eee
from the formal power series expansion  
\be
V_{\pm } = V \left (x \pm \frac 12 \hbar y ,t \right ) = \sum_{\ell =0}^\infty \frac {1}{\ell !} \frac {\partial^{\ell} V(x,t)}{\partial x^{\ell}} 
\left ( \pm \frac 12 \hbar y \right )^\ell \, . 
\ee
For instance
\be
F_2 &=&  \int_{\R} \left ( V_+ - V_- \right ) \psi_+ \overline {\psi_-} e^{-i y \xi } dy \\
&=& 2 \sum_{m =0}^\infty \frac {1}{(2m+1) !} \frac {\partial^{2m+1} V(x,t)}{\partial x^{2m+1}} \left ( 
\frac 12 \hbar \right )^{2m+1} \int_{\R} y^{2m+1} \psi_+ \overline {\psi_-} e^{-i y \xi } dy \\ 
&=& 2 \sum_{m =0}^\infty \frac {1}{(2m+1) !} \frac {\partial^{2m+1} V(x,t)}{\partial x^{2m+1}} \left ( 
\frac i2 \hbar \right )^{2m+1} \frac {\partial^{2m+1}}{\partial \xi^{2m+1}}\int_{\R}  \psi_+ \overline {\psi_-} e^{-i y \xi } dy 
\\ 
&=& 4\pi \sum_{m =0}^\infty \frac {1}{(2m+1) !} \frac {\partial^{2m+1} V(x,t)}{\partial x^{2m+1}} \left ( 
\frac i2 \hbar \right )^{2m+1} \frac {\partial^{2m+1}W^\hbar }{\partial \xi^{2m+1}} \, . 
\ee

Thus, from (\ref {Eq31}), (\ref {Eq33}) and (\ref {Eq35}) Theorem \ref {Teo_3} follows. \ Similarly, by noticing that 
\be
i \hbar \int_{\R} \left ( \dot \psi_+ \overline {\psi_-} - \psi_+ \overline {\dot \psi_-} \right ) e^{-i y \xi } dy  = 
E \int_{\R} \left ( \psi_+ \overline {\psi_-} + \psi_+ \overline {\psi_-} \right ) e^{-i y \xi } dy = 
4\pi E W^\hbar
\ee
when $\psi (x,t;\hbar )= e^{-iEt/\hbar} \varphi (x;\hbar )$, then $W^\hbar$ is independent of $t$ and from (\ref {Eq32}), (\ref {Eq34}) and (\ref {Eq36}) 
Theorem \ref {Teo_4} follows.

The proofs are thus completed. 

\subsection {Schr\"odinger equations with singular potential or with a nonlinear potential in the Wigner representation} \label {Sez3_2}

Here we consider the cases where the potential $V$ is a singular function, namely a Dirac's delta potential, or where a nonlinear potential occurs, namely 
we consider the Gross-Pitaevskii equation. \ In both cases one can write a formal equation to the Wigner representation of the wave-function. 

\subsubsection {Schr\"odinger equation with a Dirac's delta potential in the Wigner representation} \label {Sez3_2_1}

We premise the following result.

\begin {lemma} \label {lem_7}
Let $V(x) =\gamma \delta_{x_0} (x)$ be a Dirac's delta distribution supported at the point $x={x_0}$. \ Then 
\bee
F_2 &=&  
\frac {4 \gamma i}{\hbar} \int_{\R} W^\hbar \left ( x , \xi ' \right ) \sin \left [ \frac {2({x_0}-x) (\xi' -\xi)}{\hbar}\right ] d \xi' \label {Eq37} \\ 
F_4 &=& \frac {4 \gamma}{\hbar} \int_{\R} W^\hbar \left ( x , \xi ' \right ) \cos \left [ \frac {2({x_0}-x) (\xi' -\xi)}{\hbar}\right ] d \xi'  
\label {Eq38}
\eee
\end {lemma}

\begin {proof}
Indeed (let us denote $\psi (x,t;\hbar)$ by $\psi (x)$ for sake of simplicity), 
\be
F_2 &=&  \int_{\R} \left ( V_+ - V_- \right ) \psi_+ \overline {\psi_-} e^{-i y \xi } dy \\
&=& \gamma \int_{\R} \left [ \delta_{x_0} \left ( x + \frac 12 \hbar y \right ) - \delta_{x_0} \left ( x - \frac 12 \hbar y \right ) \right ] \psi \left ( x + \frac 12 \hbar y \right ) \overline {\psi \left ( x - \frac 12 \hbar y \right )}e^{-i y \xi} dy \\ 
&=& \frac {4 \gamma i}{\hbar} \Im \left [ \psi \left ( {x_0} \right ) \overline {\psi \left ( 2x - {x_0} \right )}e^{-i2 \xi ({x_0}-x)/\hbar} \right ]
\ee
and similarly
\be
F_4 
= \frac {4\gamma}{\hbar}  \Re \left [ \psi \left ( {x_0} \right ) \overline {\psi \left ( 2x - {x_0} \right )}e^{-i2 \xi ({x_0}-x)/\hbar} \right ]
\, . 
\ee

Now, recalling Lemma \ref {Lem_4} then 
\be
\psi ({x_0}) \overline {\psi (2x-{x_0})} = \int_{\R} W^\hbar \left ( x , \xi ' \right ) e^{i2({x_0}-x) \xi' /\hbar} d \xi' \, .
\ee

Therefore we can conclude that
\be
F_2 =\frac {4 \gamma i}{\hbar} \Im \left \{ \int_{\R} W^\hbar \left ( x , \xi ' \right ) e^{i2({x_0}-x) (\xi' -\xi)/\hbar} d \xi' \right \} = 
\frac {4 \gamma i}{\hbar} \int_{\R} W^\hbar \left ( x , \xi ' \right ) \sin \left [ \frac {2({x_0}-x) (\xi' -\xi)}{\hbar}\right ] d \xi'
\ee
and 
\be
F_4 = \frac {4 \gamma }{\hbar} \Re \left \{ \int_{\R} W^\hbar \left ( x , \xi ' \right ) e^{i2({x_0}-x) (\xi' -\xi)/\hbar} d \xi' \right \}
= \frac {4 \gamma }{\hbar} \int_{\R} W^\hbar \left ( x , \xi ' \right ) \cos \left [ \frac {2({x_0}-x) (\xi' -\xi)}{\hbar}\right ] d \xi' \, . 
\ee
\end {proof}

Thus the time-dependent Schr\"odinger 
\be
i \hbar \frac {\partial \psi}{\partial t} = - \hbar^2 \frac {\partial^2 \psi}{\partial x^2} + \gamma \delta_{x_0} \psi 
\ee
takes the form of integro-differential equation
\bee
\frac {\partial W^\hbar}{\partial t} = - 2 \xi \frac {\partial W^\hbar}{\partial x} + \frac {2\gamma}{\pi \hbar^2} \int_{\R} 
 W^\hbar (x,\xi') \sin \left [ \frac {2({x_0}-x)(\xi'-\xi )}{\hbar}\right ] d\xi'\label {Eq39}
\eee
and the time-independent Schr\"odinger takes the form of integro-differential equation
\bee
E W^\hbar = -  \frac {\hbar^2}{4}  \frac {\partial^2 W^\hbar}{\partial x^2}  + \xi^2 W^\hbar   
+ \frac {\gamma }{\hbar \pi } \int_{\R} W^\hbar (x,\xi') \cos \left [ \frac {2({x_0}-x)(\xi'-\xi )}{\hbar}\right ] d\xi' \label {Eq40}
\eee
with the constrain
\bee
0 = -2 \xi \frac {\partial W^\hbar}{\partial x} + \frac {2\gamma}{\pi \hbar^2} \int_{\R} 
W^\hbar (x,\xi') \sin \left [ \frac {2({x_0}-x)(\xi'-\xi )}{\hbar}\right ] d\xi'    \label {Eq41}
\eee

\begin {remark} \label {Nota7}
One can easily check that (\ref {Eq18}) satisfies (\ref {Eq40}-\ref {Eq41}).
\end {remark}

\subsubsection {Gross-Pitaevskii equation in the Wigner representation} \label {Sez3_2_2} The Gross-Pitaevskii equation has the form
\bee
\left \{
\begin {array}{l}
i \hbar \frac {\partial \psi}{\partial t} = H\psi + \nu |\psi|^2 \psi  \\ 
\psi (x,0;\hbar )=\psi_0 (x;\hbar ) 
\end {array}
\right. \ \mbox { where } \ H := - \frac {\hbar^2}{2m} \frac {\partial^2 }{\partial x^2} + V(x,t) \, , \ \nu \in \R \, . \label {Eq42}
\eee

In such a case the Wigner transform of equation (\ref {Eq42}) satisfies to the formal integro-differential equation 
\bee
 \frac {\partial W^\hbar}{\partial t} = -2 \xi \frac {\partial W^\hbar}{\partial x} +
\sum_{m=0}^{+\infty} \frac {(-1)^m \hbar^{2m} }{ 2^{2m}(2m+1)!} \frac {\partial^{2m+1} W^\hbar}{\partial \xi^{2m+1}}  \frac {\partial^{2m+1} }{\partial 
x^{2m+1} } \left [ V(x,t) + \nu    \int_{\R} W^\hbar (x,\xi', t  ) d\xi' \right ]  \label {Eq43}
\eee

By means of a straightforward (formal) calculation one obtains that (\ref  {Eq31}) becomes 
\be
2\pi \hbar i \frac {\partial W^\hbar}{\partial t} = -\hbar^2 F_1 + F_2 + F_5
\ee
where $F_1$ and $F_2$ are defined in (\ref {Eq33}) and (\ref {Eq35}). \ Concerning the term $F_5$ we make use of the shortened notation (\ref {Eq26}), thus 
\be
F_5 &=&  \nu \int_{\R} e^{-iy \xi} \left [ \left | \psi_+ \right |^2 \psi_+ \overline {\psi_- }- \psi_+ \left | \overline {\psi_- } \right |^2
\overline {\psi_- } \right ] dy \\
&=&   \nu  \int_{\R} e^{-iy \xi} Y_- (x,y,t;\hbar ) \psi_+ \overline {\psi_- } dy
\ee
where we set
\be
Y_\pm (x,y,t;\hbar ) &=&  \left | \psi \left (x + \frac 12 \hbar y ,t ; \hbar \right ) \right |^2 \pm 
\left | \overline {\psi \left (x - \frac 12 {\hbar} y ,t ;\hbar \right )} \right |^2  \\
&=& \int_{\R} \left [ W^\hbar \left (x+\frac 12 \hbar y ,\xi' ,t  \right )\pm W^\hbar \left (x-\frac 12 \hbar y ,\xi' ,t  \right ) \right ] d\xi'
\, .
\ee
Then, the formal power series expansion $ W^\hbar \left (x\pm \frac 12 \hbar y ,\xi' ,t  \right ) = \sum_{\ell=0}^\infty \frac {1}{\ell!} 
\frac {\partial^\ell W^\hbar}{\partial x^\ell} \left ( \pm \frac 12 \hbar y \right )^\ell$ yields to 
\be
Y_- (x,y,t;\hbar ) = \sum_{m=0}^{+\infty} \frac {\hbar^{2m+1} y^{2m+1}}{(2m+1)! 2^{2m}} \frac {\partial^{2m+1}}{\partial x^{2m+1}}\int_{\R} 
W^\hbar (x,\xi' ,t ) d\xi 
\ee
Therefore,
\be
&& F_5 =  \nu \int_{\R} e^{-iy \xi} 
 \sum_{m=0}^{+\infty} \frac {\hbar^{2m+1} y^{2m+1}}{(2m+1)!2^{2m}} \frac {\partial^{2m+1}}{\partial x^{2m+1}}\int_{\R} W^\hbar (x,\xi' ,t ) d \xi 
\psi_+ \overline {\psi_- } dy \\
&& =   \nu  \sum_{m=0}^{+\infty} \frac {\hbar^{2m+1} }{(2m+1)!2^{2m}} \left [ \frac {\partial^{2m+1}}{\partial x^{2m+1}} 
\int_{\R} W^\hbar (x,\xi' ,t )d\xi \right ] \int_{\R} y^{2m+1} e^{-iy \xi} \psi_+  
\overline { \psi_- }dy \\
&& = 2\pi \hbar i \nu \sum_{m=0}^{+\infty} \frac {(-1)^m \hbar^{2m} }{(2m+1)!2^{2m}} \left [ \frac {\partial^{2m+1}}{\partial x^{2m+1}} \int_{\R} W^\hbar (x,\xi' ,t  ) d\xi' \right ] \, 
\left [ \frac {\partial^{2m+1}}{\partial \xi^{2m+1}} W^\hbar (x,\xi ,t ) \right ] 
\ee
from which (\ref {Eq43}) follows. 

Similarly, the case of time-independent Gross-Pitaevskii equation can be treated. \ That is, if the potential $V(x)$ does not depend on $t$ and if 
$\psi (x,t;\hbar)=e^{-iEt/\hbar}\varphi (x;\hbar)$ is a solution to 
the time-independent Gross-Pitaevskii equation $E\varphi = H \varphi + \nu |\varphi|^2 \varphi$ then its Wigner transform $W^\hbar = \W^\hbar (\varphi )$ is a function 
independent of $t$ which satisfies the equation
\bee
EW^\hbar =-  \frac {\hbar^2}{4}  \frac {\partial^2 W^\hbar}{\partial x^2}  + 
\sum_{m =0}^\infty \frac {(-1)^m \hbar^{2m}}
{2^{2m}(2m) !}  \frac {\partial^{2m}W^\hbar }{\partial \xi^{2m}} \frac {d^{2m} }{d 
x^{2m} } \left [  V(x) + \nu    \int_{\R} W^\hbar (x,\xi' ) d\xi' \right ] \label {Eq44}
\eee
under the constrain 
\bee
 0 = -2 \xi \frac {\partial W^\hbar}{\partial x} +
\sum_{m=0}^{+\infty} \frac {(-1)^m \hbar^{2m} }{ 2^{2m}(2m+1)!} \frac {\partial^{2m+1} W^\hbar}{\partial \xi^{2m+1}} \frac {d^{2m+1} }{d 
x^{2m+1} }  \left [ V(x)+ \nu   \int_{\R} W^\hbar (x,\xi' ) d\xi' \right ]  \label {Eq45}
\eee
Indeed, the constrain (\ref {Eq45}) immediately follows from (\ref {Eq43}) and from the fact that $W^\hbar$ must be independent of $t$. \ In order to 
prove (\ref {Eq44}) one must remark that (\ref {Eq32}) becomes
\be
4\pi E W^\hbar = - \hbar^2 F_3 + F_4  + F_6
\ee
where $F_3$ and $F_4$ are defined by (\ref {Eq34}) and (\ref {Eq36}), and where 
\be
F_6 &=&  \nu  \int_{\R} e^{-iy \xi} Y_+ (x,y;\hbar ) \psi_+ \overline {\psi_- } dy \\ 
&& 
= 2\pi \nu \sum_{m=0}^{+\infty} \frac {(-1)^m \hbar^{2m-1} }{(2m)!2^{2m-1}} \left [ \frac {d^{2m}}{d x^{2m}} 
\int_{\R} W^\hbar (x,\xi' ) d\xi' \right ] \, 
\left [ \frac {\partial^{2m}}{\partial \xi^{2m}} W^\hbar (x,\xi  ) \right ] 
\ee
from which (\ref {Eq44}) follows. 

\section {Harmonic oscillator via the Wigner transform} \label {Sez4}

Here we apply the Wigner transform in case of real-valued potentials $V$ given by polynomials of second degree with respect to $x$; indeed, if $V$ is a 
polynomial with degree $r>2$ then (\ref {Eq21}) becomes a PDE quite hard to solve.

In Section \ref {Sez4_1} we consider, at first, the eigenvalue problem for the harmonic oscillator in the Wigner representation, following the results given by  
\S 3.6.1 \cite {Schleich}. \ In Section \ref {Sez4_2} we consider then the solution to the time-dependent Schr\"odinger equation in the Wigner representation 
when the potential has the form $V(x,t)=\gamma x^2 + Q(t) x $ where $\gamma \in \R$ is a fixed constant and $Q(t)$ is any function depending on time. \ In fact, the 
case where $\gamma = \gamma (t)$ depends on time may be similarly treated but we don't dwell here on such a problem. 

\subsection {Time-independent Schr\"odinger equation in the Wigner representation} \label {Sez4_1}

If we look for a stationary solution $\psi (x,t)= e^{-iEt/\hbar} \varphi (x)$, for $E\in \R$, then $W^\hbar$ is actually $t$-independent. \ In the harmonic 
oscillator model the potential $V(x) = \omega^2 x^2 $ is a polynomial independent of $t$ of second degree with respect to $x$; from 
Theorem \ref {Teo_4} then $W^\hbar \in L^2_\R (\R \times \R , dx d\xi )$ is the solution to the eigenvalue problem
\bee
E W^\hbar = - \frac {\hbar^2}{4} \frac {\partial W^\hbar}{\partial x^2} + \left [ \xi^2 + V(x) \right ] W^\hbar  - \frac {\hbar^2}{8} 
\frac {d^2 V}{d x^2} \frac {\partial^2 W^\hbar}{\partial \xi^2}
 \label {Eq46}
\eee
under the constrain
\bee
0 = - 2 \xi \frac {\partial W^\hbar}{\partial x} +  \frac {\partial V}{\partial x} \frac {\partial W^\hbar }{\partial \xi}  \, . 
\label {Eq47}
\eee 
Equation (\ref {Eq47}) implies that the solution $W^\hbar$ has the form
\be
W^\hbar (x,\xi )= f(y) \, , \ y = \xi^2 + V(x)\, , 
\ee
where $f$ is a real-valued function. \ From equation  (\ref {Eq46}) it turns out that $f$ must satisfy to the equation (where $f'$ and 
$f''$ respectively denote the first and second derivatives of $f(y)$ with respect to its argument $y$):
\be
Ef &=& \left [- \frac {\hbar^2}{4} \left ( \frac {dV}{dx} \right )^2 - \frac {\hbar^2}{2} \xi^2 \frac {d^2V}{dx^2} \right ] f''  
- \frac {\hbar^2}{2} \frac {d^2V}{dx^2} f' +  y f \\ 
&=& - \omega^2 \hbar^2 y f'' + y f -  \omega^2 \hbar^2 f' 
\ee
that is
\bee
\tau \frac {d^2v}{d\tau^2} + (b-t) \frac {dv}{d\tau} -a v =0 \label {Eq48}
\eee
where
\be
b=1 \ \mbox { and } \ a= \frac 12  - \frac 12  \frac {\mathcal E}{\sqrt {\lambda}}
\ee
and 
\be
f(y)= e^{-\tau/2} v(\tau )\, , \ \tau =2 \sqrt {\lambda} y\, , 
\lambda = \omega^{-2} \hbar^{-2} \ \mbox { and } \ {\mathcal E} = E \omega^{-2} \hbar^{-2}\, . 
\ee

The general solution to (\ref {Eq48}) is a linear combination of the two Kummer's functions $M$ and $U$ \cite {AS}:
\be
v(\tau ) = C_1 M(a,b,\tau ) + C_2 U(a,b,\tau )\, . 
\ee
Recalling also that
\be
U(a,b,\tau ) \sim \frac {\Gamma (b-1)}{\Gamma (a)} \tau^{1-b} \ \mbox { as } \ \tau \to 0^+ 
\ee
then $C_2=0$; furthermore, recalling that 
\be
M(a,b,\tau ) \sim \frac {\Gamma (b)}{\Gamma (a)} e^{\tau} \tau^{a-b} + e^{\pm i\pi a} \frac {\tau^{-a}}{\Gamma (b-a)} \ \mbox { as } \ \tau \to + \infty
\ee
then $f(y)$ belongs to $L^2 (\R^+ )$ if $-a= n \in \N$, that is  
\be
-\frac 12 + \frac 12  \frac {\mathcal E}{\sqrt {\lambda}} \in \N \ \mbox { that is } \ {\mathcal E}_n = \frac {2n+1}{\hbar \omega }\, , \ n=0,1, \ldots . 
\ee
In particular, for ${\mathcal E}={\mathcal E}_n $ then 
\be
f(y) = C_1 e^{-\tau /2} M(-n,1,\tau ) = C_1 e^{-y/\hbar \omega} n! L_n (2y/\hbar \omega)
\ee
where $L_n$ is the $n$-th Laguerre polynomial. \ Hence, 
\bee
E =E_n = \omega^2 \hbar^2 {\mathcal E}_n =  (2n+1) \omega \hbar \label {Eq49}
\eee
and 
\be
W^\hbar (x,\xi )=  C e^{- \frac {\xi^2 + \omega^2 x^2}{\hbar \omega}} L_n \left ( \frac {2}{\hbar \omega} (\xi^2 + \omega^2 x^2 ) \right )
\ee
where $C$ is a normalization constant. \ That is we have proved that the time-independent Schr\"odinger equation has a real-valued stationary 
solutions $W^\hbar \in L^2_{\R} (\R \times \R , dx d\xi)$ when the energy $E$ is given by (\ref {Eq49}). \ However, we should also check that 
$W^\hbar$ is the Wigner transform function of a pure state. \ In fact, $W^\hbar \in \D$ because it is, up to a multiplication factor, the Wigner 
transform of the function $\varphi_n \left ( \sqrt {\frac {\omega}{\h}} x \right ) $ (see Example \ref {Es4}). 

\subsection {Time-dependent Schr\"odinger equation in the Wigner representation} \label {Sez4_2}

Assume that $V(x,t) = \gamma x^2 + Q(t) x$; then, from Theorem \ref {Teo_3} it follows that $W^\hbar \in L^2_\R (\R \times \R , dx d\xi )$ is the solution 
to the Cauchy problem
\bee
\left \{ 
\begin {array}{l}
\frac  {\partial W^\hbar}{\partial t} = - 2 \xi \frac {\partial W^\hbar}{\partial x} + \left [ 2 \gamma x + Q(t) \right ]  \frac {\partial W^\hbar }{\partial \xi}  \\ 
W^\hbar (x, \xi , 0) = \left [ \W^\hbar (\psi_0 ) \right ] (x, \xi ) =W^\hbar_0 (x,\xi )
\end {array}
\right. \, . \label {Eq50}
\eee

\begin {theorem} \label {Teo_5}
 Let $\gamma >0$ and let 
 \be
 X(x,\xi ,t)&:=&  a_1 (t) x + a_2 (t)  \xi + a_3 (t)  \\
\Xi (x,\xi ,t) &:=& b_1 (t) x + b_2 (t)  \xi + b_3 (t)
\ee
where
 \be
 a_1 (t)=  \cos (2\sqrt {\gamma } t ) \, , \ a_2 (t)= - \frac {1}{\sqrt {\gamma}} \sin (2\sqrt {\gamma }t)  \, , \ a_3 (t) =- \frac {1}{\sqrt {\gamma}}
\int_0^t Q(t') \sin (2\sqrt {\gamma }t')  dt' \\
b_1 (t)= \sqrt {\gamma} \sin (2\sqrt {\gamma } t ) \, , \ b_2 (t)=\cos (2\sqrt {\gamma }t) \, , \ b_3 (t)= \int_0^t Q(t') \cos (2\sqrt {\gamma }t')  dt' \, . 
\ee
Then the solution to the Cauchy problem (\ref {Eq50}) is given by
\bee
W^\hbar (x,\xi ,t) =W^\hbar_0 \left [ X(x,\xi ,t), \Xi (x,\xi ,t ) \right ] \, .  \label {Eq51}
\eee
\end {theorem}

\begin {proof} We look for solution to equation (\ref {Eq50}) of the form
\bee
W^\hbar (x,\xi ,t ) = G\left [ f(x,\xi ,t), g(x,\xi ,t) \right ] \label {Eq52}
\eee
where $G(f,g)$ is a real-valued function and where $f$ and $g$ are linear function with respect to $x$ and $\xi$; i.e.
\be
f (x,\xi ,t)=C_1 (t) x + C_2 (t) \xi + C_3 (t) \ \mbox { and } \ g (x,\xi ,t)=C_4 (t) x + C_5 (t) \xi + C_6 (t)
\ee
Thus, by substituting (\ref {Eq52}) in (\ref {Eq50}) it turns out that $C_i (t)$, $i=1,\ldots , 6$, are solutions to the following ODEs 
\be
\left \{
\begin {array}{lcl}
\dot C_1 & =& 2 \gamma C_2 \\
\dot C_2 &=& - 2 C_1 \\
\dot C_3 &=& Q(t) C_2 \\
\dot C_4 &=& 2 \gamma C_5 \\ 
\dot C_5 &=& - 2 C_4 \\
\dot C_6 &=& Q(t) C_5
\end {array}
\right. \ \Rightarrow \ 
\left \{
\begin {array}{lcl}
C_1 & =& c_1 \cos \left (2 \sqrt {\gamma} t\right ) + c_2 \sin \left (2 \sqrt {\gamma} t\right )  \\
C_2 &=& \frac {1}{\sqrt {\gamma}} \left [ - c_1  \sin \left (2 \sqrt {\gamma} t\right ) +  c_2  \cos \left (2 \sqrt {\gamma} t\right ) \right ] \\
C_3 &=& \int_0^t Q (t')  C_2(t') dt' + c_3 \\
C_4 & =& c_4 \cos \left (2 \sqrt {\gamma} t\right ) + c_5 \sin \left (2 \sqrt {\gamma} t\right )  \\
C_5 &=& \frac {1}{\sqrt {\gamma}} \left [ - c_4 \sin \left (2 \sqrt {\gamma} t\right ) +  c_5 \cos \left (2 \sqrt {\gamma} t\right ) \right ] \\
C_6 &=& \int_0^t Q (t')  C_5(t') dt' + c_6
\end {array}
\right.
 \, ,
\ee
where $c_i$, $i=1,2,\ldots , 6$, are integration constants.

In order to satisfy to the initial condition it follows that the function $G$ is such that
\be
W^\hbar_0 (x,\xi )= G (u,v) 
\ee
where 
\be 
\left \{
\begin {array}{lcl}
u&=& f(x,\xi ,0)= c_1 x + \frac {c_2}{\sqrt {\gamma}}\xi +c_3 \\
v&=& g(x,\xi ,0)=  c_4 x + \frac {c_5}{\sqrt {\gamma}}\xi +c_6 
\end {array}
\right.
\, .
\ee
Inverting such a relation, provided that $c_1c_5 \not= c_2c_4$, it follows that 
\be 
\left \{
\begin {array}{lcl}
x &=& \frac {(u-c_3)c_5-(v-c_6)c_2}{c_1 c_5- c_2 c_4} \\
\xi &=& -\frac {(u-c_3)c_4-(v-c_6)c_1}{c_1 c_5- c_2 c_4} \sqrt {\gamma}
\end {array}
\right.
\ee
and thus
\be
G(u,v)= W^\hbar_0 \left ( \frac {(u-c_3)c_5-(v-c_6)c_2}{c_1 c_5-c_2 c_4}, -\frac {(u-c_3)c_4-(v-c_6)c_1}{c_1c_5-c_2c_4} \sqrt {\gamma} \right ) \, .
 \ee
From this fact and from (\ref {Eq52}), then Theorem \ref {Teo_5} follows. \ If $c_1c_5 = c_2c_4$ then the same result can be obtained considering the limit 
$c_1 c_5 \to c_2 c_4$. 
\end {proof}

\begin {remark} \label {Nota8}
The case where $\gamma <0$ holds true by simply recalling that $\cos (i\theta )= \cosh (\theta)$ and $\sin (i \theta ) = i \sinh (\theta )$. \ The case $\gamma =0$ 
holds true by simply taking the limit $\gamma \to 0$ in (\ref {Eq51}). 
\end {remark}

\begin {remark} \label {Nota9} 
The Hamiltonian flux $(q(t),p (t)) = S^t (x,\xi )$ associated to the Hamiltonian $h (p,q,t)= p^2 + V(q,t)$, where $V(q,t) = \gamma q^2 +  Q(t) q $, is the solution to the Hamiltonian system 
(where $2m=1$)
\bee
\left \{
\begin {array}{lcl}
\dot q &=& - 2 p = -\frac {\partial h}{\partial p} \\ 
\dot p &=& 2\gamma q +  Q(t) = \frac {\partial V}{\partial q} = \frac {\partial h}{\partial q} 
\end {array}
\right.
\ \mbox { and } \ 
\left \{
\begin {array}{lcl}
q(0) &=& x \\
p (0) &=& \xi
\end {array}
\right. \label {Eq53}
\eee
and it is given  
\be
q(t) 
&=& \cos (2\sqrt {\gamma } t) x - \frac {1}{\sqrt {\gamma}} \sin (2\sqrt {\gamma } t) \xi 
- \frac {1}{\sqrt {\gamma}} \int_0^t  Q (t') \sin \left [2\sqrt {\gamma } (t-t')\right ] dt' \\
&=& 
X(x,\xi ,t) + \left [ a_2 (t) b_3 (t) - b_2 (t) a_3 (t) - a_3 (t) \right ] \\
p (t) 
&=& \sqrt {\gamma }  \sin (2\sqrt {\gamma } t) x + \cos (2\sqrt {\gamma } t) \xi  +  \int_0^t  Q(t') \cos \left [ 2\sqrt {\gamma } (t-t') \right ] dt' \\
&=& 
\Xi (x,\xi ,t) + \left [ a_1 (t) b_3 (t) - b_1 (t) a_3 (t) - b_3 (t) \right ] \, . 
\ee
When the potential $V$ is {\bf independent} of $t$, i.e. $Q (t) \equiv constant$, then we have that 
\bee
W^\hbar (x,\xi ,t) = \left [ W^\hbar_0 \circ S^t \right ] (x,\xi )\, , \label {Eq54}
\eee
since $x(t)= X(x,\xi ,t)$ and $p(t)=\Xi (x,\xi ,t)$ in such a case, and thus (\ref {Eq54}) agrees with Theorem \ref {Teo_5}. \ We must remark that (\ref {Eq54}) 
does not hold true in general when the potential $V$ actually {\bf depends} on $t$, i.e. when $Q(t)$ is not a constant function.
\end {remark}

\section {Dynamics of a Gaussian wave-function} \label {Sez5}

Now, we apply Theorem \ref {Teo_5} to the study of the solution $\psi (x,t)$ to the Schr\"odinger equation
\bee
\left \{
\begin {array}{ll}
i \hbar \frac {\partial \psi}{\partial t} = - \hbar^2 \frac {\partial^2 \psi}{\partial x^2} + \left [ \gamma x^2 + Q(t) x \right ] \psi \\
\psi (x,0)=\psi_0 (x) 
\end {array}
\right. 
\label {Eq55}
\eee
where the initial wavefunction has a Gaussian shape:
\bee
\psi_0 (x) = \frac {1}{\sqrt [4]{\pi \hbar}} e^{-(x-a)^2/2\hbar} e^{i p_0 x /\hbar } \label {Eq56}
\eee
such that $\| \psi_0 \|_{L^2 (\R , dx)} =1$ and 
\bee
<x>^0 = \langle \psi_0 , x \psi_0 \rangle = a \ \mbox { and } \ < p >^0 = -i \hbar \langle \psi_0 , \nabla \psi_0 \rangle = p_0 \, . \label {Eq57}
\eee
Its Wigner transform is (Example \ref {Es3})
\bee
W^\hbar_0 (x,\xi ) = \left [ \W^\hbar (\psi_0 ) \right ] (x,\xi ) = \frac {1}{\pi \hbar} e^{-(x-a)^2/\hbar} e^{-(\xi - p_0)^2/\hbar} \, . \label {Eq58}
\eee

\begin {theorem}\label {Teo_6}
Let $\psi_0$ given by (\ref {Eq56}); let $a_j (t)$ and $b_j (t)$, $j=1,2,3$, defined in Theorem \ref {Teo_5}; let   
\be
A(t) &:=& a_2^2 (t) + b_2^2(t) \\
B(x,t) &:=& 2 a_2 (t) \left [ a_1(t) x +a_3 (t)-a \right ]  +   2 b_2 (t) \left [ b_1(t) x +b_3 (t)-p_0 \right ] \\
v(t)&=& -b_2(t) \left [ a_3(t)-a \right ] +a_2 (t) \left [ b_3 (t)- p_0\right ] 
\ee
then 
\bee
|\psi (x,t)|^2 = \frac {1}{ \sqrt {\pi \hbar A(t)} } e^{-\frac {(x-v(t))^2}{ \hbar A(t)}} \label {Eq59}
\eee
and 
\bee
\psi (x,t) =  \frac {1}{\sqrt [4]{\pi \hbar A(t)
}} e^{i \left [ \theta^\star  - \frac {B(x/2,t)x}{2\hbar A(t)} \right ]}
e^{- \frac {\left (x -v(t) \right )^2}{2\hbar A(t)} }\label {Eq60}
\eee
for some phase $\theta^\star$ independent of $x$.
\end {theorem}

\begin {proof}
From Theorem \ref {Teo_5} it follows that
\bee
W^\hbar (x,\xi , t) &=& \frac {1}{\pi \hbar} e^{-[X(x,\xi ,t)-a]^2/\hbar} e^{-[\Xi (x, \xi ,t) - p_0]^2/\hbar} \nonumber \\
&=& \frac {1}{\pi \hbar} e^{-\left [ A(t) \xi^2 + B(x,t) \xi + C(x,t)\right ] /\hbar } \label {Eq61}
\eee
where $A(t)$ and $B(x,t)$ has been previously defined and where 
\be
C(x,t) :=   \left [ a_1(t) x +a_3 (t)-a \right ]^2 + \left [ b_1(t) x +b_3 (t)-p_0 \right ]^2
\ee
A straightforward calculation gives that
\bee
|\psi (x,t)|^2 &=& \int_{\R} W^\hbar (x,\xi ,t ) d\xi = \frac {1}{ \sqrt {\pi \hbar A(t)} } e^{-\frac {4C(x,t) A(t)-B^2(x,t)}{4A(t) \hbar}} \nonumber 
\\
&=& \frac {1}{ \sqrt {\pi \hbar A(t)} } e^{-\frac {d_1(t) x^2 + d_2 (t) x + d_3 (t)}{A(t) \hbar}}\, , \label {Eq62}
\eee
where (let us omit the dependence on the variable $t$)
\be
d_1 &:=& (a_2 b_1 - a_1 b_2)^2 =1\\ 
d_2 &:=& 2 (a_2 b_1 - a_1 b_2)(-b_2a_3 + b_2 a+b_3 a_2 - a_2 p_0)=-2 v(t)\\
d_3 &:=& (-b_2a_3+b_2a+b_3a_2 -a_2 p_0)^2 = v^2(t)
\ee
since $a_2 b_1 - a_1 b_2=-1$. \ Hence, (\ref {Eq59}) follows.

In order to prove (\ref {Eq60}) we apply Lemma 5 with $x^\star =0$, then it follows that 
\be
\psi (2x,t) &=& {e^{i \theta^\star }} \sqrt [4] {\pi \hbar A(t)} e^{d_3 (t)/2 A(t) \hbar} \int_\R W^\hbar \left (  x , \xi ,t 
\right ) e^{i2x\xi /\hbar } d\xi \\
&=& \frac {1}{\sqrt [4]{\pi \hbar A(t)
}} e^{i \left [ \theta^\star  - \frac {B(x,t)x}{\hbar A(t)} \right ]}
e^{- \frac {1}{2\hbar A(t)} \left [ 2x -v(t) \right ]^2}
\ee
for some phase factor $\theta^\star$ independent of $x$. \ Hence (\ref {Eq60}) follows.
\end {proof}

\begin {remark} \label {Nota10}
A straightforward calculation gives that
\be
<x>^t = \langle \psi (x,t), x \psi (x,t) \rangle = v(t) \, . 
\ee
We remark that $v(t)$ coincides with the function $q(t)$ associated to the classical flux: $(q(t),p(t)) = S^t (a,p_0)$ discussed in Remark \ref {Nota9}; this fact agrees with the Ehrenfest Theorem for 
quadratic Hamiltonians \cite {BLT}.
\end {remark}

Now, we are going to apply Theorem \ref {Teo_6} to different cases in order to get an explicit expression of the probability density $|\psi (x,t )|^2$, and of the 
wavefunction $\psi (x,t)$ too in the simplest cases. \ The most simple cases (from the free problem to the harmonic oscillator) have been already known 
(see, e.g. \cite {Teta}). \ Eventually, we consider the case of the forced harmonic (when $\gamma >0$)/inverted (when $\gamma <0$) oscillator where $V(x,t)= 
\gamma x^2 + \left [ \lambda + b \cos (\Omega t)\right ] x$ for some $\gamma ,\, \lambda \, , b\, , \Omega \in \R$. 

\subsection {Free and Linear Stark potential} \label {Sez5_1}
In such a case $V(x)= \lambda x$ for some $\lambda \in \R$ and thus equation (\ref {Eq50}) takes the form 
\be
\frac {\partial W^\hbar}{\partial t} = - 2 \xi \frac {\partial W^\hbar}{\partial x} + \lambda \frac {\partial W^\hbar}{\partial \xi}
\ee
From (\ref {Eq51}) the real-valued general solution to this equation has the form (where $Q(t) \equiv \lambda$ and where we take the limit $\gamma \to 0$)
\be
W^\hbar (x,\xi , t) = W^\hbar_0 
\left ( x-\lambda t^2 - 2\xi t ,  \xi +\lambda t \right ) = \frac {1}{\pi \hbar} e^{-(x-\lambda t^2 - 2 \xi t - a)^2/\hbar} e^{-(\xi + \lambda t -p_0)^2/\hbar}\, . 
\ee
Since
\be
A=4t^2+1 \, , \ B=-4t(x-a-\lambda t^2)+2\lambda t -2p_0 \ \mbox { and } \ v=a+2p_0 t- \lambda t^2   
\ee
and then 
\be
|\psi (x,t )|^2 =  \frac {1}{\sqrt {\pi \hbar} \sqrt {4t^2+1} }e^{-\frac {(x+\lambda t^2 -a-2p_0 t)^2}{\hbar (4t^2+1)}}
\ee
and
\be
\psi (x,t ) =  \frac {1}{\sqrt[4] {\pi \hbar} \sqrt[4] {4t^2+1} } e^{i \left [ \theta^\star + \frac {(x-2 \lambda t^2- 2a)xt+p_0x- \lambda t x}{\hbar (4t^2+1)} \right ] }
e^{-\frac {(x+\lambda t^2 -a-2p_0 t)^2}{2\hbar (4t^2+1)}}
\ee
for some $\theta^\star$ depending on $t$. \ If $\lambda =0$ then the solution obtained agrees with well known results (see Example \ref {Es5}).

\subsection {Harmonic oscillator} \label {Sez5_2}
In such a case $V(x)=\omega^2 x^2$ for some $\omega \not=0 $ and thus equation (\ref {Eq50}) takes the form
\bee
\frac {\partial W^\hbar}{\partial t} = - 2 \xi \frac {\partial W^\hbar}{\partial x} + 2 \omega^2 x \frac {\partial W^\hbar}{\partial \xi} \label {Eq63}
\eee

From (\ref {Eq51}) the real-valued general solution to this equation has the form 
\be
W^\hbar (x,\xi ,t) &=&  
\frac {1}{\pi \hbar} \exp 
\left [ -(A (t)\xi^2 + B(x,t)\xi + C(x,t))/\hbar \right ]
\ee
where 
\be
A (t) &=& \frac {1}{\omega^2} \sin^2 2\omega t  + \cos^2 2\omega t \\ 
B(x,t) &=&  - \frac {2}\omega (x \cos 2 \omega t -a) \sin 2 \omega t + 2(\omega x \sin 2 \omega t - p_0  ) \cos 2 \omega t  \\ 
C (x,t)&=& + (x \cos 2 \omega t -a)^2 + (\omega x \sin 2 \omega t -p_0 )^2
\ee
and
\be
v(t)= a\cos (2\omega t) + \frac {p_0}{\omega} \sin (2\omega t)\, , 
\ee
then
\be
|\psi (x,t)|^2 = \left [ \pi \hbar \left ( \frac {1}{\omega^2} \sin^2 (2\omega t) + \cos^2 (2\omega t) \right ) \right ]^{-1/2} \exp \left [ -  
\frac {\left [ x-v(t)\right ]^2}{\hbar \left ( \frac {1}{\omega^2} \sin^2 (2\omega t) + \cos^2 (2\omega t) \right )} \right ] \, . 
\ee

\begin {remark} \label {Nota11}
The inverted oscillator model may be similarly treated by simply recalling that $\cos (i\theta )= \cosh (\theta)$ and $\sin (i \theta)= i \sinh (\theta)$.
\end {remark}

\subsection {Forced oscillator} \label {Sez5_3} Let 
\be
V(x,t)= \gamma x^2 + Q(t)x \ \mbox { where } \ Q(t) = \left [ \lambda + b \cos (\Omega t)\right ] 
\ee
for some $\gamma , \, b\, , \lambda \in \R$; it is the potential of a forced harmonic oscillator (when $\gamma >0$) or of a forced inverted oscillator 
(when $\gamma <0$). \ We focus our attention here to the inverted oscillator, where $\gamma <0$; the harmonic forced oscillator where $\gamma >0$ may be similarly treated, one has only 
to separately consider the non-resonant case, where $4\gamma \not= \Omega^2$, and the resonant case, where $4\gamma = \Omega^2$.

Equation (\ref {Eq50}) takes the form
\bee
\dot W^\hbar = - 2 \xi \frac {\partial W^\hbar}{\partial x} + \left [ 2 \gamma x + \lambda + b \cos (\Omega t) \right ] 
\frac {\partial W^\hbar}{\partial \xi} \, , \label {Eq64}
\eee
and it has solution (\ref {Eq51}). \ In particular, let $\gamma = - \omega^2$, $\omega>0$, then a straightforward calculation gives that 
 \be
 a_1 (t)= b_2 (t)= \cosh (2\omega t ) \, ,\ a_2 (t)= - \frac {1}{\omega} \sinh (2\omega t)  \, , \ b_1 (t)= \omega \sinh (2\omega t )
\ee
and 
\be
a_3 (t) 
 &=& \frac {\lambda}{2 \omega^2} \left [ 1- \cosh (2\omega t) \right ]   -b 
\frac {\Omega \sinh (2\omega t) \sin (\Omega t) + 2 \omega \left [\cosh (2\omega t) \cos (\Omega t) -1\right ] }{\omega (4\omega^2 + \Omega^2)} 
 \\
b_3 (t)&=& \frac {\lambda}{2\omega} \sinh (2\omega t) +b \frac { 2 \omega \cos (\Omega t) \sinh(2\omega t) +  \Omega \sin (\Omega t) \cosh (2\omega t) 
}{(4\omega^2 + \Omega^2)} 
\, . 
\ee
In such a case
\be
v(t) = -\frac {\lambda \left [ \cosh (2\omega t)-1\right ] -2a\omega^2 \cosh(2\omega t)-2
\omega p_0 \sinh(2\omega t)}{2\omega^2 } +  \frac {2b(\cos(\Omega t)-\cosh(2\omega t))}{(\Omega^2+ 4 \omega^2)}
\ee

\subsection {Tunnel effect for the inverted oscillator} \label {Sez5_4} The question we consider is quite simple \cite {BV,Barton,HSADV}: suppose that the quantum 
wave-function has a Gaussian shape such that $< x >^0 = a $ and $< p >^0 = p_0 $; then, assuming that $a<0$ and $p_0>0$, we would compute the probability 
\be
P(t)=\int_{-\infty}^0 |\psi (x,t)|^2 dx 
\ee
to find the quantum particle in the left-hand semi-axis $x<0$ when $t$ goes to infinity for different values of the energy. \ In fact, $P(t)$ in the framework 
of quantum mechanics represents the probability to find the particle in the interval $(-\infty ,0)$ at the instant $t$. 

Initially we consider the case on the inverted oscillator, and then the case of the forced inverted oscillator.

\subsubsection {Tunnel effect for the undriven inverted oscillator} \label {Sez5_4_1} Let us consider an initial wavefunction of the shape (\ref {Eq56}) where we 
assume, for argument's sake, again $2m=1$; then 
\be
E_q = \langle \psi_0 , H \psi_0 \rangle = \frac 12 (1-\omega^2) \hbar +p_0^2-\omega^2 a^2  
\ee
is the energy in quantum mechanics. \ Let $v(t)$ be the function introduced in Theorem \ref {Teo_6}; it describes the classical motion of the particle associated to the 
initial condition $v(0)=a$ and $\dot v (0) =\frac {p_0}{m}=2p_0$; the energy in classical mechanics is given by
\be
E_c = \frac {p^2}{2m} + V(x) = p_0^2 - \omega^2 a^2 \, .
\ee

When $p_0 = p_{crit} :=|\omega a| $ then the energy $E_c=0$ is equal to the top of the potential $V(x)=-\omega^2 x^2$. \ If 
$p_0 < p_{crit}$ then the energy $E_c<0$ is less than the top of the potential and we classically expect that a particle, 
initially at $x=a<0$ and moving forward, exhibits an inversion motion. \ Finally, If 
$p_0 > p_{crit}$ then the energy $E_c>0$ is bigger than the top of the potential and we classically expect that a particle, 
initially moving forward from $x=a<0$, passes the barrier and keeps going without reversing the motion. 

Classically we expect that $P(t)$, which initially is close to the value $1$ (for $\hbar$ small enough), will be always bigger that $\frac 12$ and goes to $1$ when $t$ goes to $+\infty$ if the energy is less that the barrier top; on the other hand, when the energy is bigger than the barrier top we expect that $P(t)$ takes the value $\frac 12$ at some instant $t$ and then it goes to $0$ when $t$ goes to $+\infty$. 

Because of the tunnel effect such a picture is quite different from a quantum mechanical point of view. \ Indeed, since $|\psi (x,t)|^2$ is given 
by (\ref {Eq59}) then a straightforward calculation gives that
\bee
P(t) = \frac 12 \left [ 1- \mbox {erf} \left ( \frac {v(t)}{\sqrt {\hbar A(t)} }\right ) \right ] \label {Eq65}
\eee
In particular, it easily follows that  
\bee
P(t)= \frac 12 - \frac 12 \mbox {erf} \left [ 
\frac {p_0\sinh (2\omega t) + a \omega \cosh (2\omega t)}{\sqrt {\hbar} \sqrt {\sinh^2 (2\omega t) + \omega^2 \cosh^2 (2\omega t)}}
\right ]
\eee
and
\bee
P_\infty:= \lim_{t\to + \infty} 
P(t)= \frac 12 - \frac 12 \mbox {erf} \left [ 
\frac {p_0+ a \omega}{\sqrt {\hbar} \sqrt {1 + \omega^2}}
\right ]\, .
\eee
One must remarks that we recover the classical picture in the limit $\hbar \to 0^+$  (see Figure \ref {Fig1}). 

We can collect these comments as follows recalling that $<x>^t = v(t)$: 

\begin {itemize}
 
\item [-] if $p_0 <p_{crit}$ then 
\be
\lim_{t\to + \infty} <x>^t = - \infty 
\ee
and $\frac 12 < P_\infty < 1$;

\item [-] if $p_0 =p_{crit}$ then 
\be
\lim_{t\to + \infty} <x>^t = 0 
\ee
and $P_\infty =\frac 12 $;

\item [-] if $p_0 > p_{crit}$ then 
\be
\lim_{t\to + \infty} <x>^t = +\infty  
\ee
and $0< P_\infty <\frac 12 $.
\end {itemize}

\begin{figure}
\begin{center}
\includegraphics[height=8cm,width=8cm]{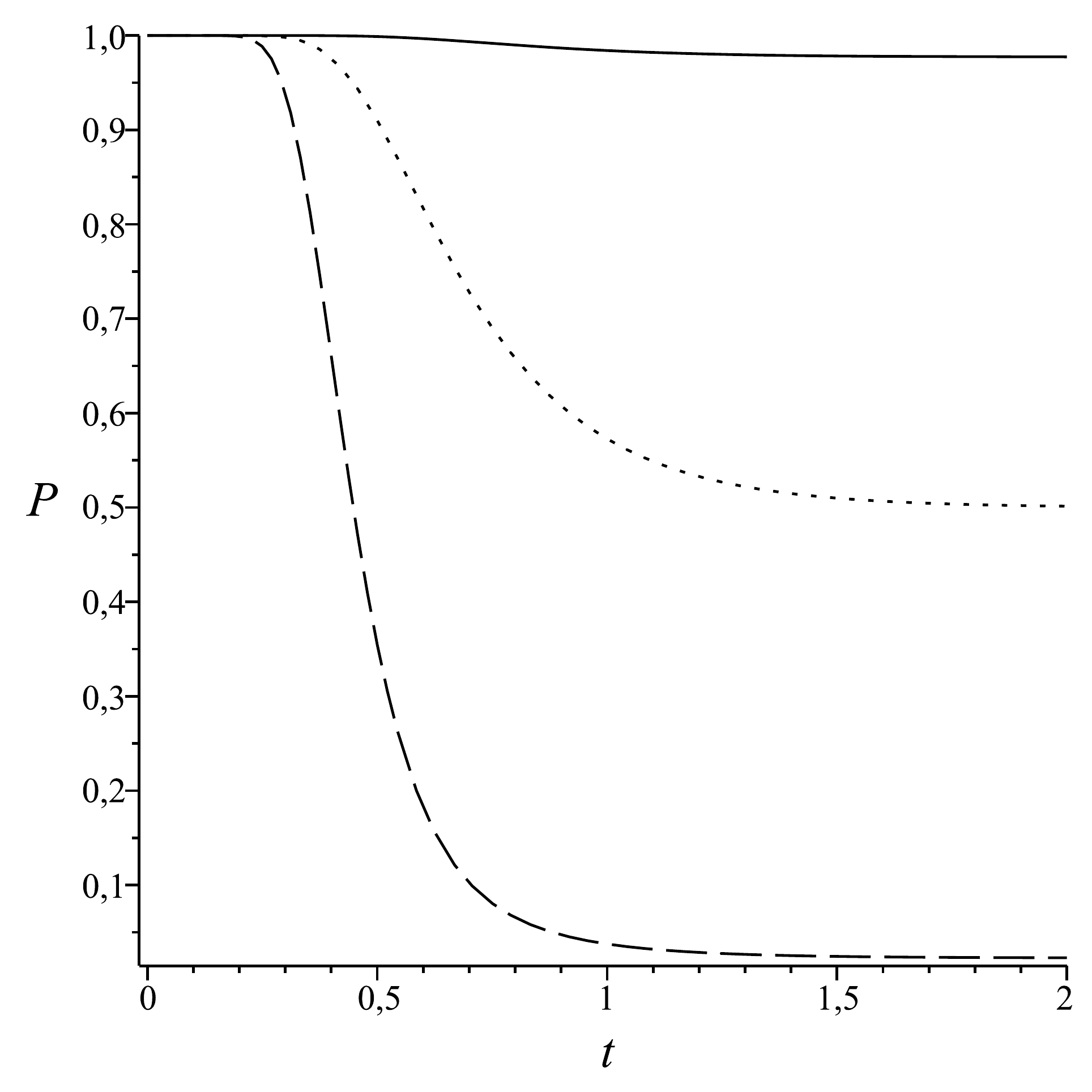}
\caption{Plot of the function $P(t)$ for $p_0<p_{crit}$ (full line), $p_0=p_{crit}$ (dot line) and $p_0>p_{crit}$ (broken line); for argument's sake we 
choose $a=-5$, $\omega=1$ and $\hbar =1$.}
\label{Fig1}
\end{center}
\end{figure}

\subsubsection {Tunnel effect for the driven inverted oscillator} \label {Sez5_4_2} If $V(x,t)= - \omega^2 x^2 + \left [\lambda + b \cos (\Omega t) \right ] x$ then one can prove that
\be
P_\infty &=& \frac 12 - \frac 12 \mbox {erf} \left [  \frac {\left ( -\lambda +2a\omega^2 + 2 p_0 \omega \right ) (\Omega^2 + 4 \omega^2)- 4 b \omega^2}{2\sqrt {\hbar} \omega \sqrt {1+\omega^2} (\Omega^2 + 4 \omega^2) }\right ] \\
&=& \frac 12 - \frac 12 \mbox {erf} \left [ -\frac {\lambda}{2\sqrt {\hbar} \omega \sqrt {1+\omega^2}} + \frac {a\omega + p_0}{\sqrt {\hbar}  \sqrt {1+\omega^2}}- \frac {2b\omega}{(\Omega^2 + 4 \omega^2)\sqrt {\hbar}  \sqrt {1+\omega^2}} \right ]
\ee
In such a case the critical value for $p_0$ takes the form
\be
p_{crit} = \frac {\lambda (\Omega^2 + 4 \omega^2 )+ 4 b \omega^2 - 2 \omega^2 a (\Omega^2 + 4 \omega^2) }{2\omega (\Omega^2 + 4 \omega^2)}\, , 
\ee
i.e.
\be
\mbox { if } \ p_0 
\left \{
\begin {array}{ll}
 < p_{crit} & \ \mbox { then } \ \frac 12 < P_\infty < 1 \\
 = p_{crit} & \ \mbox { then } \ P_\infty = \frac 12 \\
 > p_{crit} & \ \mbox { then } \ 0 < P_\infty < \frac 12
 \end {array}
 \right. 
 \, . 
 \ee

\appendix

\section {Proof of (\ref {Eq16})} \label {App0}
Let 
\be
W^\h (x,\xi ) = \frac {1}{2\pi R \xi} \sin \left [ \frac {2\xi}{\h} (R-|x|) \right ] \chi_{[-R,+R]}(x) \, . 
\ee
Then 
\be
I:= \| W^\hbar (\cdot , \cdot )\|_{L^1 (\R \times \R , dx d\xi )} = \frac {1}{2\pi R } \int_{\R} d\xi \int_{-R}^{R} dx \frac {1}{|\xi |} \left | \sin \left [ \frac {2\xi}{\h} (R-|x|) \right ] \right | \, . 
\ee
We are going to prove that this integral diverges. \ To this end let us assume, for argument's sake, that $R=1$ and $\h =2$; thus we have to compute the integral
\be
 I= \frac {1}{\pi } \int_{0}^{+\infty} d\xi \int_{0}^{R} dx \frac {1}{\xi } \left | \sin \left [ \xi (1-x) \right ] \right | \, . 
\ee
For any fixed $\xi \ge \frac 34 \pi$ and any positive integer $n=0,1,2\ldots ,$ let
\be
A_{n,\xi} : = \left \{ x \in [0,+1] \ : \ a_n:= 1- \frac {\left (n + \frac 34 \right ) \pi }{\xi } \le x \le b_n:= 1- \frac {\left (n + \frac 14 \right ) \pi }{\xi } \right \} \, , 
\ee
where $a_n \ge 0$  provided that 
\be
0 \le n \le N(\xi ):= 
\left 
\lceil \frac {\xi-\frac 34 \pi }{\pi} \right \rceil 
\ee
where $\lceil y\rceil $ denote the integer part of $y$ and where $N(\xi ) \ge 0$ for $\xi \ge \frac 34 \pi$. \ Furthermore, the measure of the interval 
$A_{n,\xi}$ is $|A_{n,\xi}|=\frac {\pi}{2\xi}$ for any $n \le N(\xi )$ and 
\be
\left | \sin \left [ \xi (1-x) \right ] \right | \ge \frac {1}{\sqrt {2}} \, ,\ \forall x \in A_{n,\xi }\, .
\ee
Let $ B_\xi := \cup_{n=0}^{N(\xi )} A_{n,\xi } $  with measure $ 
|B_\xi | =\left [ N(\xi ) +1 \right ] \frac {\pi}{2\xi} \sim \frac 12$ 
for large $\xi$. \ 
Hence, the integral
\be
I \ge \frac {1}{\sqrt {2} \pi} \int_{\frac 34 \pi}^{+\infty } d\xi \frac {1}{\xi} \int_{B_{\xi}} d x =  \frac {1}{\sqrt {2} \pi} \int_{\frac 34 \pi}^{+\infty } 
d\xi \frac {1}{\xi} |B_\xi |
\ee
diverges.

\section {Proof of (\ref {Eq19})} \label {App1} 
Let $\alpha =-\frac 12 \gamma >0$ and let
\be
\psi (x,t;\hbar ) = 
\frac {\sqrt {\alpha }}{2\hbar} \mbox {sech} \left ( {\frac {\alpha}{2\hbar^2}} x \right )e^{i \alpha^2 t/4 \hbar^3}\, ;
\ee
its Wigner's transform is 
\be 
W^\hbar &=& \frac {1}{2\pi } \int_{\R} e^{-iy\xi } \psi \left ( x + \frac 12 \hbar y  , t\right ) \overline {\psi \left ( x - \frac 12 \hbar y ,t \right ) }
dy \\ 
&=& \frac {\alpha}{8 \pi \hbar^2} \int_{\R} e^{-iy\xi } \mbox {sech} \left [ \frac {\alpha}{2\hbar^2} \left ( x + \frac 12 \hbar y  \right ) \right ] 
\mbox {sech} \left [ \frac {\alpha}{2\hbar^2} \left ( x - \frac 12 \hbar y  \right ) \right ] 
dy 
\ee
If we set $a=\frac {\alpha}{2\hbar^2}$ and $b= \frac {\alpha}{4\hbar}$ we have to deal with the integral
\be
I &=& \int_{\R} e^{-i y \xi } \mbox {sech} (ax+by) \mbox {sech} (ax-by) dy \\ 
&=& 
4\int_{0}^{+\infty} \cos ( y \xi ) \frac {1}{\mbox {cos} (\beta ) + \mbox {cosh} (\delta y)} dy \\ 
&=& 4 \pi \delta^{-1} \mbox {csc} (\beta ) \sinh (\delta^{-1} \beta \xi ) \mbox {csch} (\delta^{-1} \pi \xi ) \\
&=& \frac {2 \pi}{b} \frac {\sin (ax \xi /b )}{\sinh (2ax ) \sinh ( \pi \xi /2b)}
\ee
where $\beta = 2ax i$ and $\delta = 2b$, and where we make use of the integral transform (6) \S 1.9 page 30 \cite {B}. \ Thus 
\be 
W^\hbar  =  \frac {\alpha}{8 \pi \hbar^2} \frac {2 \pi}{b} \frac {\sin (ax \xi /b )}{\sinh (2ax ) \sinh ( \pi \xi /2b)} 
= \frac {1}{\hbar} \frac {\sin \left ( x \xi \frac {2}{\hbar} \right )}
{\sinh \left ( \frac {\alpha}{\hbar^2}x \right ) \sinh \left ( \pi \xi \frac {2\hbar}{\alpha} \right )} 
\ee
proving (\ref {Eq19}).

\end {document}